\documentclass[11pt]{article}
\usepackage[lmargin=1.0in,rmargin=1.0in,bottom=1.0in,top=1.0in,twoside=False]{geometry}

\usepackage{fullpage,amssymb,amsmath,amsthm}
\usepackage{graphicx}
\usepackage{enumerate}
\usepackage[T1]{fontenc}
\usepackage{lmodern}

 \usepackage{xcolor}
 \usepackage{mathtools}
\usepackage{microtype}
\usepackage{amsfonts}
\usepackage{comment}
\usepackage[english]{babel}
\usepackage{mathrsfs}
\usepackage{trimspaces}
\usepackage{nccfoots}
\usepackage{setspace}
\usepackage[absolute]{textpos}
\usepackage{thmtools}
\usepackage{thm-restate}
\usepackage{todonotes}

\definecolor{blue}{rgb}{0.1,0.2,0.5}
\definecolor{brown}{rgb}{0.6,0.6,0.2}
\usepackage[ocgcolorlinks, linkcolor={blue}, citecolor={brown}]{hyperref}

\usepackage{comment}
\usepackage{cleveref}
\usepackage{latexsym}
\usepackage{xspace}

\renewcommand{\S}{\mathcal{S}}
\def\dd{\hbox{-}}

\newtheorem{theorem}{Theorem}[section]
\newtheorem{corollary}[theorem]{Corollary}
\newtheorem{lemma}[theorem]{Lemma}
\newtheorem{observation}[theorem]{Observation}
\newtheorem{claim}{Claim}[theorem]
\newenvironment{claimproof}[1][\unskip]{\noindent {\emph{Proof of Claim #1.\space}}}{\hfill$\triangleleft$\smallskip}
\newcommand*\samethanks[1][\value{footnote}]{\footnotemark[#1]}

\DeclareMathOperator{\anchor}{anchor}
\newcommand{\N}{\mathbb{N}}

\newcommand{\R}{\mathbb{R}}
\newcommand{\cS}{\mathcal{S}}

\newcommand{\cD}{\mathcal{D}}

\newcommand{\cA}{\mathcal{A}}

\newcommand{\dist}{\mathsf{dist}}

\DeclareMathOperator{\tw}{tw}
\newcommand{\Oh}{\mathcal{O}}

\newcommand{\Heavy}{\mathsf{Heavy}}

\newcommand{\NP}{\textsf{NP}\xspace}
\newcommand{\PP}{\textsf{P}\xspace}

\newcounter{tbox}

\newcommand{\sta}[1]{\vspace*{0.3cm}\refstepcounter{tbox}\noindent{ \parbox{\textwidth}{(\thetbox) \emph{#1}}}\vspace*{0.3cm}}
\newcommand{\vsp}{\vspace*{3mm}}

\let\originalleft\left
\let\originalright\right
\renewcommand{\left}{\mathopen{}\mathclose\bgroup\originalleft}
\renewcommand{\right}{\aftergroup\egroup\originalright}

\renewcommand{\leq}{\leqslant}
\renewcommand{\geq}{\geqslant}

\newcommand{\pot}{\mathrm{potato}}

\newcommand{\boundary}{\mathrm{boundary}}
\newcommand{\wei}{\mathfrak{w}}

\usepackage{todonotes}

\usepackage[sort&compress,numbers]{natbib}
\usepackage[displaymath, mathlines,running]{lineno}
\begin{document}

\title{Polynomial-time algorithm for Maximum Independent Set in bounded-degree graphs with no long induced claws%
\footnote{The preliminary version of the paper was presented on the conference SODA 2022~\cite{ACDR21}.}
}
\author{
  Tara Abrishami\thanks{Supported by NSF Grant DMS-1763817 and
    NSF-EPSRC Grant DMS-2120644.}\\
  Princeton University, Princeton, NJ, USA
  \\
  \\
Maria Chudnovsky\samethanks \\
Princeton University, Princeton, NJ, USA
\\
\\
Cemil Dibek\thanks{Supported by NSF Grant DMS-1763817}\\
Princeton University, Princeton, NJ, USA
\\
\\
Pawe{\l} Rz\k{a}\.{z}ewski\thanks{Supported by Polish National Science Centre grant no. 2018/31/D/ST6/00062.}\\
Warsaw University of Technology, Poland/University of Warsaw, Poland
}
\date{}

\begin{titlepage}
\maketitle
\begin{abstract}
For graphs $G$ and $H$, we say that $G$ is $H$-free if it does not contain $H$ as an induced subgraph.
Already in the early 1980s Alekseev observed that if $H$ is connected, then the \textsc{Max Weight Independent Set} problem (MWIS) remains \textsc{NP}-hard in $H$-free graphs, unless $H$ is a path or a subdivided claw, i.e., a graph obtained from the three-leaf star by subdividing each edge some number of times (possibly zero). Since then determining the complexity of MWIS in these remaining cases is one of the most important problems in algorithmic graph theory.

A general belief is that the problem is polynomial-time solvable, which is witnessed by algorithmic results for graphs excluding some small paths or subdivided claws. A more conclusive evidence was given by the recent breakthrough result by Gartland and Lokshtanov [FOCS 2020]: They proved that  MWIS can be solved in quasipolynomial time in $H$-free graphs, where $H$ is any fixed path.
If $H$ is an arbitrary subdivided claw, we know much less: The problem admits a QPTAS and a subexponential-time algorithm [Chudnovsky et al., SODA 2019].

In this paper we make an important step towards solving the problem by showing that for any subdivided claw $H$, MWIS is polynomial-time solvable in $H$-free graphs of bounded degree.
\end{abstract}
\end{titlepage}
%\newpage

\section{Introduction}
For a graph $G=(V,E)$, a set $X \subseteq V$ is \emph{independent} (or \emph{stable}) if there is no edge in $E$ with both endpoints in $X$.
In the \textsc{Max Independent Set} problem (MIS) we ask to compute $\alpha(G)$, i.e., the size of a largest independent set in the instance graph $G$. The \textsc{Max Weight Independent Set} problem (MWIS) is a generalization of MIS where each vertex $v$ is assigned a positive weight $\wei(v)$, and we ask for maximum $W$ such that the instance graph has an independent set $X$ with $W = \sum_{v \in X} \wei(v)$.

MIS (and MWIS) is a ``canonical'' hard problem: it was one of the first problems shown to be \NP-hard~\cite{Karp1972}, it is notoriously hard to approximate~\cite{Hastad96cliqueis,Khot2006}, and it is \textsc{W}[1]-hard~\cite{Cyganetal}. Many of these hardness results hold even if we restrict input instances to some natural graph classes~\cite{GAREY1976237,DBLP:conf/isaac/BonnetBTW19,DBLP:conf/wg/DvorakFRR20}.

In this work we are interested in graph classes defined by forbidding certain substructures.
For graphs $G$ and $H$, we say that $G$ is \emph{$H$-free} if it does not contain $H$ as an induced subgraph.
For simplicity, we will assume that $H$ is connected.

The complexity study of MWIS in $H$-free graphs dates back to the early 1980s and the work of Alekseev~\cite{alekseev1982effect},
who observed that for most graphs $H$ the problem remains \NP-hard.
Indeed, let us discuss the hard cases.
First, MIS (and thus MWIS) is \NP-hard in subcubic graphs~\cite{GAREY1976237}, which are $H$-free whenever $H$ has a vertex of degree at least 4.
For the remaining cases we use the so-called \emph{Poljak construction}~\cite{Po74}: if $G'$ is obtained from $G$ by subdividing one edge twice, then $\alpha(G') = \alpha(G)+1$. Thus, if $G^p$ denotes the graph obtained from $G$ by subdividing each edge exactly $2p$ times, then $\alpha(G^p) = \alpha(G) + p \cdot |E(G)|$.
Now observe that if $H$ has a cycle or two vertices of degree three, then $G^{|V(H)|}$ is $H$-free. Consequently, for such graph $H$, MIS is \NP-hard in $H$-free graphs.
Let us point out that the above hardness reductions imply that the problem cannot even be solved in subexponential time, unless the Exponential-Time Hypothesis (ETH) fails.

Summing up, the only connected graphs $H$ for which we may hope for a polynomial-time algorithm for MIS in $H$-free graphs are paths and \emph{subdivided claws} (or \emph{long claws}), where a subdivided claw is a graph obtained from the three-leaf star $K_{1,3}$ (called the \emph{claw}) by subdividing each edge some number of times (possibly 0).
In what follows, for $t \geq 1$, by $P_t$ we denote the $t$-vertex path.
For integers $a,b,c \geq 1$, by $S_{a,b,c}$ we denote the subdivided claw, where the edges of  $K_{1,3}$ were subdivided $a-1$, $b-1$, and $c-1$ times, respectively. Alternatively, we may think of $S_{a,b,c}$ as the graph obtained from paths $P_{a+1}, P_{b+1}$, and $P_{c+1}$ by identifying one endvertex of each path. We also extend this notation by allowing $a=0$: then $S_{0,b,c}$ is the path $P_{b+c+1}$.

The complexity of MWIS in $H$-free graphs when $H$ is a path or a subdivided claw remains one of the most challenging and important problems in algorithmic graph theory. Despite significant attention received from the graph theory and the theoretical computer science communities, only partial results are known. Let us discuss them.

First, consider the case that $H = P_t$ for some $t$.
Since $P_4$-free graphs, also known as \emph{cographs}, have very rigid structure (in particular, they have clique-width at most 2), the polynomial-time algorithm for MWIS in this class of graphs is rather simple~\cite{CORNEIL1981163}.
However, even for $P_5$-free graphs the situation is much more complicated.
The existence of a polynomial-time algorithm for this case was a long-standing open problem that was finally resolved in the affirmative in 2014 by Lokshtanov, Vatshelle, and Villanger~\cite{DBLP:conf/soda/LokshantovVV14} using the framework of \emph{potential maximal cliques}.
Later, using the same approach but with significantly more technical effort, Grzesik, Klimo\v{s}ova, Pilipczuk, and Pilipczuk~\cite{DBLP:conf/soda/GrzesikKPP19} obtained a polynomial-time algorithm for $P_6$-free graphs.
The case of $P_7$-free graphs remains open. 

However, some interesting algorithmic results can be obtained if we relax our notion of an efficient algorithm.
First, it was shown by Bacs\'o et al.~\cite{DBLP:journals/algorithmica/BacsoLMPTL19} that for every fixed $t$, MWIS can be solved in \emph{subexponential} time $2^{\Oh(\sqrt{n \log n})}$ for $P_t$-free graphs (by $n$ we always denote the number of vertices of the instance graph).
Another subexponential-time algorithm, with worse running time, was obtained independently by Brause~\cite{DBLP:journals/dam/Brause17}.
While these results do not rule out the possibility that the problem is \NP-hard,
let us recall that, assuming the ETH, subexponential algorithms for MWIS in $H$-free graphs cannot exist if $H$ is not a path or a subdivided claw.
Later, Chudnovsky et al.~\cite{DBLP:conf/soda/ChudnovskyPPT20,DBLP:journals/corr/abs-1907-04585} showed that for every fixed $t$, the problem admits a QPTAS in $P_t$-free graphs.
Finally, a very recent breakthrough result by Gartland and Lokshtanov~\cite{gartlandpK} shows that for every fixed $t$, the problem can be solved in $P_t$-free graphs in \emph{quasipolynomial time} $n^{\Oh(\log^3 n)}$; there is also a slightly simpler algorithm by Pilipczuk, Pilipczuk, and Rz\k{a}\.zewski~\cite{DBLP:conf/sosa/PilipczukPR21} with running time $n^{\Oh(\log^2n)}$.
Note that this means that if for some $t$, MWIS is \NP-hard for $P_t$-free graphs, then \emph{all problems} in \NP can be solved in quasipolynomial time. While this does not yet imply that \PP = \NP, it still seems rather unlikely, according to our current understanding of  complexity theory.

The case when $H$ is a subdivided claw is much less understood.
The simplest subdivided claw is the claw $S_{1,1,1}=K_{1,3}$. Claw-free graphs appear to be closely related to \emph{line graphs}~\cite{Claw5} and thus a polynomial-time algorithm for MWIS in claw-free graphs can be obtained by a modification of the well-known augmenting path approach for finding a maximum-weight matching~\cite{SBIHI198053,MINTY1980284} (i.e., a maximum-weight independent set in a line graph).
We highlight the close relation of claw-free graphs and line graphs, as it will play an important role in our paper.
The next smallest subdivided claw is the \emph{fork}, i.e., $S_{2,1,1}$. A polynomial-time algorithm for MIS in fork-free graphs was obtained by Alekseev~\cite{ALEKSEEV20043}. Later, the algorithm was extended to the MWIS problem by Lozin and Milani\v{c}~\cite{DBLP:journals/jda/LozinM08}.
The existence of polynomial-time algorithms in the next simplest cases, i.e., $H = S_{3,1,1}$ and $H=S_{2,2,1}$, is wide open. The existence of a polynomial-time algorithm for MWIS in $H$-free graphs when every connected component of $H$ is a claw was obtained by Brandstädt and Mosca \cite{mwis-claw}. 

Again, some interesting results can be obtained if we look beyond polynomial-time algorithms.
Chudnovsky et al.~\cite{DBLP:conf/soda/ChudnovskyPPT20,DBLP:journals/corr/abs-1907-04585} proved that for every subdivided claw $H$, the MWIS problem in $H$-free graphs admits a QPTAS and a subexponential-time algorithm working in time $n^{\Oh(n^{8/9})}$. We point out that the arguments used for the case when $H$ is a subdivided claw are significantly more complicated and technically involved than their counterparts for $P_t$-free graphs.

More tractability results can be obtained if we put some additional restrictions on the instance graph, e.g., by forbidding more induced subgraphs~\cite{DBLP:journals/endm/LeBS15,p71,p67,p65,p73,p72}. A slightly different direction was considered by Lozin, Milani\v{c}, and Purcell~\cite{DBLP:journals/gc/LozinMP14}, who proved that for every fixed $t$, MWIS is polynomial-time solvable in \emph{subcubic} $S_{t,t,1}$-free graphs. Later, Lozin, Monnot, and Ries~\cite{DBLP:journals/jda/LozinMR15} showed a polynomial time algorithm for MWIS in subcubic $S_{2,2,2}$-free graphs. Finally, Harutyunyan et al.~\cite{DBLP:journals/tcs/HarutyunyanLLM20} generalized both these results by providing a polynomial-time algorithm for MWIS in subcubic $S_{2,t,t}$-free graphs, for any fixed $t$.

We remark that the case when $H$ is a subdivided claw is the only case where the restriction to bounded degree graphs leads to an interesting problem. Indeed, the already mentioned hardness reduction of Alekseev~\cite{alekseev1982effect} shows that if $H$ is not a path nor a subdivided claw, then MIS in \NP-hard even in \emph{subcubic} $H$-free graphs. On the other hand, connected $P_t$-free graphs of bounded degree are of constant size and thus of little interest.

In this work, we continue the study of the complexity of MWIS in $S_{a,b,c}$-free graphs of bounded degree.
As our main result, we show the following theorem. 

\begin{restatable}{theorem}{thmalgo}
\label{thm:algorithm}
For every pair $t, \Delta$ of fixed positive integers,
given an $n$-vertex $S_{t,t,t}$-free graph $G$ with maximum degree at most $\Delta$,
equipped with a weight function $\wei : V(G) \to \N$,
in time polynomial in $n$ we can find a maximum-weight independent set in $G$.
\end{restatable}

Note that our result works for all excluded subdivided claws: the graph $S_{a,b,c}$ is an induced subgraph of $S_{t,t,t}$, where $t = \max(a,b,c)$, and thus $S_{a,b,c}$-free graphs form a subclass of $S_{t,t,t}$-free graphs.

Furthermore, it follows from the work of Lozin et al.~\cite[Theorem~2]{DBLP:journals/jda/LozinMR15} (see also~\citep{DBLP:journals/tcs/HarutyunyanLLM20}) that if $H$ is disconnected, say with connected components $H_1,H_2,\ldots,H_p$,
and for each $i \in [p]$, MWIS  is polynomial-time solvable in $H_i$-free graphs of bounded maximum degree,
then MWIS is polynomial-time solvable in $H$-free graph of bounded maximum degree.
This yields the following generalization of Theorem~\ref{thm:algorithm}.

\begin{corollary}
\label{cor:disconnected}
Let $H$ be a fixed graph, whose every component is a subdivided claw, and let $\Delta$ be a fixed integer.
Given an $n$-vertex $H$-free graph $G$ with maximum degree at most $\Delta$,
equipped with a weight function $\wei : V(G) \to \N$,
in time polynomial in $n$ we can find a maximum-weight independent set in $G$.
\end{corollary}

\subsection{Outline of the proof}
\label{subsec:outline} 

For a graph $G$, a set $Z \subseteq V(G)$ is called \emph{constricted} if there is no induced tree in $G$ that contains $Z$.
The starting point of our algorithm is a result of Chudnovsky and Seymour~\cite{3-in-a-tree} that states that if a graph $G$ contains a constricted set of size 3, then $G$ admits a certain structure, called an \emph{extended strip decomposition}.
This (very roughly) means that $G$ is similar in structure to the line graph of another graph $H$, where $G$ breaks into \emph{atoms} that correspond to vertices, edges, and triangles of $H$.
One can then attempt to use this similarity to solve the MWIS problem.
The problem with this approach is that it only works when no atom contains ``almost all'' the vertices of the graph. 

It is not difficult to see that by deleting a set $X \subseteq V(G)$ of bounded size from an $S_{t,t,t}$-free graph $G$ of bounded degree,
it is possible to obtain  a graph with a constricted set of size 3.
One can then try to make use of the extended strip decomposition of $G - X$, by showing  that such a decomposition can be found where all atoms have roughly equal sizes.
While this is a simple idea, it has remained out of reach for years. In this paper, combining the ``central bag method'' recently developed by Abrishami et al.~\cite{wallpaper} with  an in-depth analysis of the process through which extended strip decompositions are constructed,
we have finally been able to show that an extended strip decomposition where all atoms are small can indeed be found, unless $G$ contains a bounded-size balanced separator. Note that in the latter case MWIS can be solved in $G$ by an easy recursion.
On a very high level, we first identify an ``important'' claw-free  induced subgraph $\beta$ of $G$ and show that $\beta$ admits
the required decomposition, and then, using certain monotonicity properties, we extend the decomposition to the whole graph.

In what follows we describe our approach in more detail.
Let $t \geq 1$ and $\Delta \geq 3$ be fixed integers and let $G$ be an $n$-vertex $S_{t,t,t}$-free graph with maximum degree at most $\Delta$.

The proof consists of five steps and the first four of them are purely graph-theoretic.
As mentioned above, in these steps we work under an additional assumption that for some $c = c(t, \Delta) < 1$, the graph $G$ has no \emph{$c$-balanced separator} of constant size,
where a set $X \subseteq V(G)$ is \emph{$c$-balanced} if every component of $G-X$ has at most $c \cdot n$ vertices.

\paragraph{Step 1. Finding a central bag decomposition.}
In the first step, we invoke the recent result of Abrishami et al.~\cite{wallpaper} which provides a certain hierarchical decomposition of $G$, called the \emph{central bag decomposition}.
More precisely, it gives us a sequence $\beta_0, \beta_1, \ldots, \beta_k$ of connected subgraphs of $G$, such that
\begin{itemize}
\itemsep0em
\item $\beta_0 = G$ and $\beta_k$ is claw-free,
\item for each $i \in [k]$, the graph $\beta_i$ is an induced subgraph of $\beta_{i-1}$,
\item for each $i \in [k]$, every component of $\beta_{i-1} - \beta_i$ is small and its neighborhood in $\beta_i$ has bounded size and simple structure,
\item for each $i \in [k]$, the graph $\beta_i$ has no balanced separator of bounded size.
\end{itemize}
For brevity, let us denote $\beta := \beta_k$ and call it the \emph{central bag}.
We can imagine that $\beta$ is a ``core'' of $G$, and $G$ can be obtained from $\beta$ by reversing the central bag decomposition. In other words, we can start with $\beta$ and then enlarge it by iteratively adding small subgraphs that have simple interaction with the core, until we obtain the whole graph $G$. This step is described in Section~\ref{sec:bagdecomposition}.

\paragraph{Step 2. Finding a strip decomposition of the central bag.}
In the second step we analyze the central bag $\beta$.
The structure of claw-free graphs is quite well understood, due to the work of Chudnovsky and Seymour~\cite{Claw5}.
In particular, their results imply that every claw-free graph $F$  admits a \emph{strip decomposition}.
More precisely, there is some graph $H$ and a function $\eta$ mapping edges of $H$ to sets of vertices of $F$,
such that
\begin{itemize}
\itemsep0em
\item the sets $\eta(e)$ for $e \in E(H)$, called \emph{atoms}\footnote{Actually, an atom is a certain subset of $\eta(e)$, but in this high-level exposition let us not go into the detail.}, form a partition of $V(F)$,
\item for every edge $xy \in E(F)$, either $x,y \in \eta(e)$ for some $e \in E(H)$, or $x \in \eta(e)$ and $y \in \eta(f)$, where $e,f \in E(H)$ are distinct and $e$ and $f$ share a vertex,
\item for distinct $e, f \in E(H)$, such that $e$ and $f$ share a vertex, the interaction between sets $\eta(e)$ and $\eta(f)$ is well-structured.
\end{itemize}
We carefully analyze the strip decomposition obtained for $\beta$. We show that our assumptions on $\beta$,
in particular that $\beta$ has bounded degree and has no balanced separator of bounded size,
imply that  we can obtain a strip decomposition $(H,\eta)$ with useful properties. In particular, each atom of $(H,\eta)$ is small.

However, in order to proceed, we need $(H,\eta)$ to have a slightly more rigid structure.
We modify $G$ into a graph $G'$ by adding three special vertices $v_1,v_2,v_3$, which are simplicial in $G'$ (i.e., the neighborhood of each of them is a clique). Next we obtain $G''$ by removing from $G'$ a carefully chosen set $X \subseteq V(G')$ of bounded size.
We ensure the following properties of $G''$:
\begin{itemize}
\itemsep0em
\item $G''$ has maximum degree at most $\Delta$,
\item $ V(\beta) \cup \{v_1,v_2,v_3\}$ is disjoint from $X$ and it induces a connected subgraph $\beta''$ of $G''$,
\item the set $\{v_1,v_2,v_3\}$ is constricted in $G''$.
\end{itemize}
We also modify $(H,\eta)$ into a strip decomposition $(H'',\eta'')$ of $\beta''$, making sure that all atoms are still small and $(H'',\eta'')$ is \emph{tame}, which means that it satisfies certain additional technical conditions concerning the vertices $v_1,v_2,v_3$.

The property that $\{v_1,v_2,v_3\}$ is constricted allows us to use the tools developed by Chudnovsky and Seymour~\cite{3-in-a-tree} in their  solution to the famous \emph{three-in-a-tree} problem (see \cite{detect-odd, near-linear, detect-long-odd} for recent important developments involving the three-in-a-tree problem).
By the already-mentioned result it follows that every graph with a constricted three-element set admits the so-called \emph{extended strip decomposition}, which is a generalization of strip decompositions of claw-free graphs.
An extended strip decomposition of a graph $F$ is again a pair $(H,\eta)$, where $H$ is a graph and $\eta$ is a function whose domain consists of all edges of $H$ (as it was in the case of strip decompositions), and also all vertices and triangles of $H$. In particular, in addition to edge atoms, we also have vertex atoms and triangle atoms (where a {\em triangle} is a clique of size three).
The sets given by $\eta$ form a partition of $V(F)$ and the interactions between vertices from distinct sets are restricted. In particular, they are \emph{local} in $H$.

Let us point out that the decomposition $(H'',\eta'')$ of $\beta''$ is in particular an extended strip decomposition.
Furthermore, as all vertex atoms and triangle atoms are empty, we still have the property that all atoms are small.

So the outcome of this step is a set $X$ of bounded size that was removed from a graph, and a tame extended strip decomposition $(H'',\eta'')$ of the central bag (with three additional vertices added), where every atom is small.
This step is described in Section~\ref{sec:centralbag}.

\paragraph{Step 3. Extending the strip decomposition of the central bag.}
Recall that $\beta_k,\beta_{k-1},\ldots,\beta_0$ is the central bag decomposition of $G$.
Additionally, for all $i \in [k]$, the graph $\beta_i$ is an induced subgraph of $\beta_{i-1}$ and $X$ is disjoint from $\beta_k$.

The main result of this step is the following statement.
Suppose that for some $i \in [k]$, we are given a tame extended strip decomposition $(H_i,\eta_i)$ of the connected component of $\beta_i - X$,
containing $\beta_k$, such that each atom of $(H_i,\eta_i)$ is small.
Then in polynomial time we can find a tame extended strip decomposition $(H_{i-1},\eta_{i-1})$ of the connected component of $\beta_{i-1} - X$,
containing $\beta_k$, such that each atom of $(H_{i-1},\eta_{i-1})$ is small.

The proof combines the techniques developed by Chudnovsky and Seymour~\cite{3-in-a-tree} and deep structural analysis of the central bag decomposition obtained in Step~1. This step is described in Section~\ref{sec:extending}.

\paragraph{Step 4. Obtaining the extended strip decomposition of $G$.}
Using the outcome of Step~2 as the base case and the statement obtained in Step~3 as an inductive step, we can prove that the component of $G-X$ containing the central bag admits an extended strip decomposition whose every atom is small.
Then, using the properties of the central bag decomposition and the fact that $X$ is of bounded size,
we can further extend this decomposition so that it covers other components of $G-X$.

The following theorem, proven in Section~\ref{sec:combinatorial}, encapsulates the outcome of the first four steps of the proof.
It is the main technical contribution of our paper.

\begin{restatable}{theorem}{thmcomb}
\label{thm:combinatorial}
For every pair $t,\Delta$ of fixed positive integers, there exist $c \in [\frac{1}{2}, 1)$ and integers $d$, $z$ and $p$ with the following property.
Given an $n$-vertex $S_{t,t,t}$-free graph $G$ with maximum degree at most $\Delta$, either $G$ has a $c$-balanced  separator of size at most $d$, or in time $\Oh(n^{p})$ we can find a set $X \subseteq V(G)$ of size at most $z$
and an extended strip decomposition $(H,\eta)$ of $G-X$, where $H$ has at most $n$ vertices and every atom of $(H,\eta)$ has at most $\frac{1}{10\Delta} \cdot n$ vertices.
\end{restatable}

We note that a similar statement was obtained by Chudnovsky et al.~\cite[Lemma 6.5]{DBLP:journals/corr/abs-1907-04585}: They also obtain a set $X$ and an extended strip decomposition of $G-X$, where the atoms are ``small'' (although this is not measured in the number of vertices, but the total weight).
However, the main difference is that set $X$ given by  Chudnovsky et al.~\cite[Lemma 6.5]{DBLP:journals/corr/abs-1907-04585} is not of bounded size, even if the maximum degree is bounded. Thus it is not useful for our application.

\paragraph{Step 5. Solving MWIS by a reduction to \textsc{Maximum Weight Matching}.}
Finally, in the last step we use Theorem~\ref{thm:combinatorial} to prove Theorem~\ref{thm:algorithm}.
The algorithm is recursive and consists of two main steps. We use notation from Theorem~\ref{thm:combinatorial}.

First, suppose that $G$ has a $c$-balanced separator $X$ of size at most $d_{t,\Delta}$.
Then we can exhaustively guess the intersection of a fixed optimum solution with $X$ and solve the problem for each component of $G-X$ independently. The running time is polynomial, as each component is of size at most $c \cdot n$ for a constant $c<1$.

In the other case, if no such separator exists, Theorem~\ref{thm:algorithm} gives us the set $X$ of bounded size and an extended strip decomposition $(H,\eta)$ of $G-X$, where every atom is small.
We follow the idea of Chudnovsky et al.~\cite{DBLP:conf/soda/ChudnovskyPPT20,DBLP:journals/corr/abs-1907-04585}.
We define certain induced subgraphs of $G$ that we call \emph{particles}.%
\footnote{These subgraphs are called \emph{atoms} in~\cite{DBLP:conf/soda/ChudnovskyPPT20,DBLP:journals/corr/abs-1907-04585}, but we decided to use another name to avoid confusion with atoms used in Steps~2, 3, and 4.}
These subgraphs are defined with respect to $(H,\eta)$, in particular every atom of $(H,\eta)$ is a particle, but there are also particles consisting of several (but still constant number of) atoms. The fact that each atom has at most $\frac{1}{10\Delta} \; n$ vertices implies that each particle has at most $0.5n$ vertices.

Chudnovsky et al.~\cite{DBLP:conf/soda/ChudnovskyPPT20,DBLP:journals/corr/abs-1907-04585} showed that solving MWIS can be reduced to
\begin{itemize}
\itemsep0em
\item solving MWIS for all particles,
\item solving \textsc{Maximum Weight Matching} on an instance graph obtained by a simple modification of $H$.
\end{itemize}
Thus we can exhaustively guess the intersection of a fixed optimum solution with $X$,
solve the problem for each particle (here we use the fact that they are small),
and then combine the results in polynomial time solving using one of the standard algorithms for finding maximum weight matchings~\cite{edmonds_1965,DBLP:journals/mp/GoldbergK04}.
As all particles are small and the number of vertices of $H$ is at most $n$, the running time is polynomial in $n$.

\section{Background and definitions}

Let $G = (V(G), E(G))$ be a graph.
For a set $X \subseteq V(G)$, by $G[X]$ we denote the subgraph of $G$ induced by $X$ and $G \setminus X = G[V(G) \setminus X]$.
However, if it does not lead to confusion, we will often use induced subgraphs and their vertex sets interchangeably.
For sets $X, Y \subseteq V(G)$, we say that $X$ is {\em complete to} $Y$ if for every $x \in X$ and $y \in Y$, it holds that $xy \in E(G)$. We say that $X$ is {\em anticomplete to} $Y$ if for every $x \in X$ and $y \in Y$, it holds that $xy \not \in E(G)$. 

A {\em path} is a graph $P$ with vertex set $V(P) = \{v_1, v_2, \hdots, v_k\}$ and edge set $E(P) = \{v_iv_{i+1}: 1 \leq i \leq k-1\}$. We denote a path $P$ by $p_1 \dd p_2 \dd \hdots \dd p_k$. The {\em length} of a path $P$ is the number of edges in $P$. For $X, Y \subseteq V(G)$, we say that a path $P = p_1\dd \hdots \dd p_k$ with $P \subseteq V(G)$ is a {\em path from $X$ to $Y$} if $P \cap X = \{p_1\}$ and $P \cap Y = \{p_k\}$. The {\em distance between $X$ and $Y$ in $G$} is the length of a shortest path from $X$ to $Y$ in $G$. We denote the distance between $X$ and $Y$ by $\dist_G(X, Y)$. If $x, y \in V(G)$, then we define the {\em distance between $x$ and $y$}, denoted $\dist_G(x, y)$, as $\dist_G(x, y) = \dist_G(\{x\}, \{y\})$. 

Let $v \in V(G)$. By $N_G(v)$ we denote the set of vertices in $G$ adjacent to $v$, and $N_G[v] = \{v\} \bigcup N_G(v)$. For $S \subseteq V(G)$, we denote $N_G[S]=\bigcup_{v\in S} N_G[v]$ and $N_G(S) =N_G[S]\setminus S$. By $N^d_G[S]$ we denote the set of vertices at distance at most $d$ from $S$ in $G$. If the graph $G$ is clear from the context, we simplify the notation by dropping the subscript. By $\Delta(G)$ we denote the maximum degree of a vertex in $G$. 

A function $w$ is a {\em weight function on $V(G)$} if $w: V(G) \to \R$. If $w$ is a weight function on $V(G)$ and $X \subseteq V(G)$, we define $w(X) = \sum_{x \in X} w(x)$.  Let $w: V(G) \to [0, 1]$ be a weight function on $V(G)$ with $w(G) = 1$. A set $X \subseteq V(G)$ is a {\em $(w, c)$-balanced separator of $G$} if $w(D) < c$ for every connected component of $G \setminus X$. A \emph{$c$-balanced separator} is a $(w, c)$-balanced separator for the weight function $w$ on $V(G)$ given by $w(v) = \frac{1}{|V(G)|}$ for all $v \in V(G)$.
Balanced separators are related to a graph parameter called treewidth. Roughly speaking, the treewidth is a measure of how ``tree-like'' a graph is; we omit a precise definition here as it is not used in the paper.
Instead, we will work with balanced separators of bounded size; it appears that their existence is closely related to the treewidth of a graph.

\begin{lemma}[\cite{HarveyWood}]
\label{lem:tw_to_bsp}
Let $G$ be a graph.
\begin{enumerate}
\itemsep0em
\item  If the treewidth of $G$ is at most $k$, then $G$ has a $(w, c)$-balanced separator of size at most $k+1$ for every weight function $w$ with $w(G) = 1$ and for every $c \in [\frac{1}{2}, 1)$. 
\item If for some $c \in [\frac{1}{2}, 1)$ it holds that $G$ has a $(w, c)$-balanced separator of size $k$ for every weight function $w$ with $w(G) = 1$, then the treewidth of $G$ is at most $\frac{1}{1-c} \; k$.
\end{enumerate}
\end{lemma}

%For $S \subseteq V(G)$ and positive integer $k$, a set $X \subseteq V(G)$ is a {\em $(k, S, c)^*$-separator} of $G$ if $|X| \leq k$ and if $|D \cap S| \leq cn$ for every connected component $D$ of $G - X$. The {\em separation number $\sep^*_c(G)$}
\subsection{Separations and central bags} \label{sec:bagdecomposition}
A {\em separation} of a graph $G$ is a triple $(A, C, B)$ of pairwise-disjoint vertex subsets of $G$ such that $A \cup C \cup B = V(G)$ and $A$ is anticomplete to $B$. If $S = (A, C, B)$ is a separation, we write $A(S) = A$, $C(S) = C$, and $B(S) = B$. In \cite{wallpaper}, separations are used to construct a sequence of iterated decompositions of a graph. In this section, we summarize the key results from \cite{wallpaper}, which are integral to the proof of our main combinatorial result. For the remainder of the paper, we assume that if $(A, C, B)$ is a separation, then $C$ is non-empty. 

%We say that a separation $S = (A, C, B)$ is {\em $d$-bounded} if $C$ is $d$-bounded. 
Let $c \in [\frac{1}{2}, 1)$ and let $d$ be a positive integer, let $G$ be a graph with maximum degree $\Delta$, and let $w: V(G) \to [0, 1]$ be a weight function on $V(G)$ with $w(G) = 1$. Suppose $G$ has no $(w, c)$-balanced separator of size at most $d$. Let $S = (A, C, B)$ be a separation of $G$ such that $|C| \leq d$. Then, it follows that either $w(A) > c$ or $w(B) > c$, otherwise $C$ would be $(w, c)$-balanced separator of $G$ of size at most $d$. For the remainder of the paper, if $G$ has no $(w, c)$-balanced separator of size at most $d$ and $S = (A, C, B)$ is such that $|C| \leq d$, we assume by convention that $w(B) > c$ and $w(A) < 1-c$.  

Two separations $S_1 = (A_1, C_1, B_1)$ and $S_2 = (A_2, C_2, B_2)$ are {\em $A$-loosely non-crossing} if $A_1 \cap C_2 = A_2 \cap C_1 = \emptyset$. A {\em sequence of separations} is an ordered collection of separations. A sequence $\S$ of separations is {\em $A$-loosely laminar} if $S_1$ and $S_2$ are $A$-loosely noncrossing for every distinct $S_1, S_2 \in \S$. Two separations $S_1 = (A_1, C_1, B_1)$ and $S_2 = (A_2, C_2, B_2)$ are {\em non-crossing} if (possibly exchanging the roles of $A_1$ and $B_1$, and of
$A_2$ and $B_2$) $A_1 \cap C_2 = \emptyset$, $A_2 \cap C_1 = \emptyset$, and $A_1 \cap A_2 = \emptyset$. A sequence $\S$ of separations is {\em laminar} if $S_1$ and $S_2$ are non-crossing for every distinct $S_1, S_2 \in \S$. 

Let $G$ be a graph with no $d$-bounded $(w, c)$-balanced separator, and let $\S$ be an $A$-loosely laminar sequence of $d$-bounded separations of $G$. The {\em central bag for $\S$}, denoted $\beta_\S$, is 
$$\beta_\S = \bigcap_{S \in \S} \left(C(S) \cup B(S)\right).$$

%For every sequence $\S$ of separations, we
%A {\em sequence} of separations is an ordered set of separations. A sequence $\S$ of separations is {\em $A$-loosely laminar} if for every distinct $S_1, S_2 \in \S$, we have $A(S_1) \cap C(S_2) = \emptyset$ and $A(S_2) \cap C(S_1) = \emptyset$, and {\em strongly laminar} if for every distinct $S_1, S_2 \in \S$, we have $C(S_1) \cap C(S_2) = \emptyset$. 
We equip every $A$-loosely laminar sequence $\S$ of separations with an {\em anchor map}, which is a map $\anchor_\S: \S \to V(G)$ such that $\anchor_\S(S) \in C(S)$ for every $S \in \S$. We use the anchor map to define a weight function $w_\S$ on $\beta_\S$. Let $\S = (S_1, \hdots, S_k)$, and let $w^*(A(S_i)) = w(A(S_i) \setminus \bigcup_{1 \leq j < i} A(S_j))$. The {\em weight function on $\beta_\S$}, denoted $w_\S$, is defined as follows: $w_\S(v) = w(v) + \sum_{S \in \anchor_\S^{-1}(v)} w^*(A(S))$ for all $v \in \beta_\S$.  Several properties of central bags are summarized in the next lemma. 

\begin{lemma}[Abrishami et al.~\cite{wallpaper}]\label{lemma:central_bags}
Let $c \in [\frac{1}{2}, 1)$ and let $d, \Delta$ be positive integers. Let $G$ be a graph with maximum degree $\Delta$, let $w: V(G) \to [0, 1]$ be a weight function on $V(G)$ with $w(G) = 1$, and suppose $G$ has no $(w, c)$-balanced separator of size at most $d$. Let $\S$ be an $A$-loosely laminar sequence of $d$-bounded separations of $G$, let $\beta_\S$ be the central bag for $\S$, and let $w_\S$ be the weight function on $\beta_\S$. Then, the following properties hold: 
\begin{enumerate}[(i)]
\itemsep0em
    \item $C(S) \subseteq \beta_\S$ for all $S \in \S$; 
    
    \item For every component $D$ of $G \setminus \beta_\S$, there exists $S \in \S$ such that $D \subseteq A(S)$; 
    
    \item $w_\S(\beta_\S) = 1$; and
    
    \item $w(N[A(S)]) \leq w_\S(C(S))$ for all $S \in \S$.
    
%    \item $\beta_\S$ has no $(d-1)$-bounded $(w_\S, c)$-balanced separator 
\end{enumerate}
\end{lemma}

If $S$ is a separation and $H$ is an induced subgraph of $G$, we define $S \cap H$ as the separation of $H$ given by $(A(S) \cap H, C(S) \cap H, B(S) \cap H)$. If $\S$ is a sequence of separations and $H$ is an induced subgraph of $G$, we define $\S \cap H$ as $\{S \cap H \mid S \in \S\}$.

The {\em dimension} of a sequence $\S$ of separations is the minimum number of laminar sequences with union $\S$. A sequence $\S$ of separations is {\em $(a, t)$-good} if $|\anchor_\S^{-1}(v)| \leq a$ for all $v \in V(G)$ and $C(S)$ has diameter at most $t$ for all $S \in \S$. 

In this paper, we are interested in $S_{t,t,t}$-free graphs with bounded degree. 
Let $t_1 \geq 0$ and $t_2, t_3 \geq 1$. Recall that the graph $S_{t_1, t_2, t_3}$ consists of a vertex $v$ and three paths $P^1, P^2, P^3$, with one end $v$, such that $P^1 \setminus \{v\}$, $P^2 \setminus \{v\}$, and $P^3 \setminus \{v\}$ are pairwise disjoint and anticomplete to each other, and $|P^i| = t_i+1$. The vertex $v$ is called the {\em root} of $S_{t_1, t_2, t_3}$. 

The following theorem summarizes the application of central bags to $S_{t,t,t}$-free graphs.   

%An $A$-loosely laminar sequence $\S$ of separations is {\em $(a, t)$-good} if for every $v \in V(G)$, $|\anchor_\S^{-1}(v)| \leq a$ and for every $S \in \S$, the diameter of $C(S)$ is at most $t$. 

%Recall that we are working with graphs $G$ with maximum degree $\Delta$ and no $d$-bounded $(w, c)$-balanced separator for every fixed $c \in [\frac{1}{2}, 1)$ and $d \geq 1$, where $w$ is the uniform weight function on $G$. 

%In \cite{wallpaper}, they prove the following: 

\begin{theorem}[Abrishami et al.~\cite{wallpaper}]
\label{thm:claw-free-bag}
Let $c \in [\frac{1}{2}, 1)$ and let $\Delta, t, d$ be positive integers with $d \geq (3t+1)\Delta(1 + \Delta + \hdots + \Delta^t)^{3t(\Delta^{7t+4}+1)}$. Let $G$ be a connected $S_{t, t, t}$-free graph with maximum degree $\Delta$ and no $(w, c)$-balanced separator of size at most $d$, where $w(v) = \frac{1}{|V(G)|}$ for all $v \in V(G)$. Then, we can find in polynomial time a sequence $\beta_0, \beta_1, \hdots, \beta_k, \beta_{k+1}$ of graphs, where $\beta_{k+1} \subseteq \beta_k \subseteq \beta_{k-1} \subseteq \hdots \subseteq \beta_0 = G$, such that the following hold: 
\begin{itemize}
\itemsep0em
    \item $k \leq 3t(\Delta^{7t+4}+1)$;
    
    \item $\beta_i$ is connected for all $0 \leq i \leq k+1$;
    
    \item For all $1 \leq i \leq k+1$, there exists a weight function $w_i$ on $\beta_i$, with $w_i(\beta_i) = 1$;
    
    \item For all $1 \leq i \leq k+1$, there is a sequence of separations $\S_i$ of $\beta_{i-1}$, such that:
    \begin{itemize} 
    
        \item $\S_i$ is $A$-loosely laminar; 
        \item $w_{i-1}(A(S)) < 1-c$ for all $S \in \S_i$; 
        
        \item $\beta_i$ is the central bag for $\S_i \cap \beta_{i-1}$ and $w_i$ is the weight function on $\beta_i$; 
        
        \item $C(S)$ has diameter at most $2(t+1)$ and size at most $(3t+1)\Delta$ for all $S \in \S_i$; and
        
        %\item $C(S) \subseteq \beta_i$ for all $S \in \S_i$, 
        
        %\item For every distinct $S_1, S_2 \in \S_i$, we have $C(S_1) \cap C(S_2) = \emptyset$, and
        
        \item For every vertex $v \in V(G)$, the set $\{S \in \S_i \mid v \in C(S)\}$ has size at most $2^{\Delta}$.
        
        %\item Every component of $\beta_{i-1} \setminus \beta_i$ is contained in $A(S)$ for some $S \in \S_i$. 
    \end{itemize}

%    \item For all $1 \leq i \leq k+1$, every component $D$ of $\beta_{i-1} \setminus \beta_i$ satisfies the following: 
%    \begin{itemize}
%        \item $w_{i-1}(D) < 1-c$;
        
%        \item $|N(D)| \leq (3t+1)\Delta$; and
        
%        \item There exists $S \in \S_i$ such that $w_{i-1}(N[D]) \leq w_i(C(S))$. 
        
%    \end{itemize}
    
    \item For all $0 \leq i \leq k+1$, it holds that $\beta_{i}$ has no $(w_{i}, c)$-balanced separator of size at most $d(1 + \Delta + \hdots + \Delta^t)^{-i}$; 
    
        \item $\beta_{k+1}$ is claw-free and has no clique cutset.

\end{itemize}
\end{theorem}
\begin{proof}
 Let $t_1 \geq 0$ and let $t_2, t_3 \geq 1$. Let $G'$ be a connected $S_{t_1+2, t_2, t_3}$-free graph with maximum degree $\Delta$ and no $(w, c)$-balanced separator of size $d$. Let $t' = \max(t_1, t_2, t_3)$. For a graph $X$, let $\mathcal{X} = \{X' \subseteq V(G') \mid X' \text{ is isomorphic to $X$}\}$, and let $\S_X$ be the {\em $X$-covering sequence for $G'$}, as defined in \cite{wallpaper} before Lemma 2.7. In particular, $C(S) \subseteq N[H]$ for some $H$ isomorphic to $X$ in $G$.  
Let $\anchor_{\S_X}(S)$ be the root of $C(S)$ for every $S \in \S_X$. 

\sta{\label{bounded-dimension} Every vertex $v \in V(G')$ is the root of at most $\Delta^{3t'}$
claws $S_{t_1,t_2,t_3}$.}

Let $v \in V(G')$. Since $v$ has degree at most $\Delta$, there are at most ${\Delta \choose 3} \leq \Delta^3$ ways to choose three distinct neighbors $a, b, c$ of $v$. There are at most $\Delta^{t'-1}$ ways to construct paths of $G$ of length $t'-1$ with one endpoint $a$ (or $b$ or $c$). Therefore, there are at most $\Delta^{3 + 3(t'-1)} = \Delta^{3t'}$ claws of $G$ with root $v$ and paths of length at most $t'$. This proves \eqref{bounded-dimension}. \vsp 

By \eqref{bounded-dimension} and the fact that $S_{t_1, t_2, t_3}$ has diameter $2t'+2$, it follows that $\S_{S_{t_1, t_2, t_3}}$ is $(\Delta^{3t'}, 2t'+2)$-good. Therefore, by \cite[Lemma 2.7]{wallpaper}, $\dim(\S_{S_{t_1, t_2, t_3}}) \leq \Delta^{7t'+4} + 1$. 

%Let $G'$ be a connected $S_{t_1+2, t_2, t_3}$-free graph with maximum degree $\Delta$ and no $(w, c)$-balanced separator of size $d$. 
Let $k_1 = \Delta^{7t'+4} + 1$. Now, from \cite[Lemma 2.5]{wallpaper}, we deduce that we can find in polynomial time a sequence $\S_1, \hdots, \S_{k_1}$ of $A$-loosely laminar separations and a sequence $G = \beta_0, \beta_1, \hdots, \beta_{k_1}$ of induced subgraphs of $G$, satisfying the following: 
\begin{itemize}
\itemsep0em
    \item $\beta_i$ is connected for $0 \leq i \leq k_1$;
    
    \item For all $1 \leq i \leq k_1$, there exists a weight function $w_i$ on $\beta_i$ with $w_i(\beta_i) = 1$;
    
    \item $\beta_i$ is the central bag for $\S_i \cap \beta_{i-1}$ for $1 \leq i \leq k_1$;
    
    \item  $\beta_i$ has no $(w_i, c)$-balanced separator of size at most $d(1 + \Delta + \hdots + \Delta^{t'})^{-i}$ for $1 \leq i \leq k_1$. 
    
    \item For all $S \in \S_i$, it holds that $C(S) \subseteq N[H]$ for some $H \subseteq \beta_i$ isomorphic to $S_{t_1, t_2, t_3}$. 
\end{itemize}
Further, by the proof of \cite[ Lemma 2.7]{wallpaper}, $C(S) \cap C(S') = \emptyset$ for all $S, S' \in \S_i$. Finally, by \cite[Lemmas~4.1 and~2.9]{wallpaper}, it holds that $\beta_{k_1}$ does not have $S_{t_1 + 1, t_2, t_3} + K_1$ as an induced subgraph. If there exists $H \subseteq \beta_{k_1}$ such that $H$ is an $S_{t_1 +1, t_2, t_3}$ in $\beta_{k_1}$, it follows that $\beta_{k_1} \subseteq N[H]$, so $\beta_{k_1}$ has a $(w_{k_1}, c)$-balanced separator of size $(3t'+1)\Delta$, a contradiction. Therefore, $\beta_{k_1}$ is $S_{t_1 + 1, t_2, t_3}$-free. 

By iteratively applying the procedure described thus far $3t'$ times, we obtain a sequence $\beta_0, \hdots, \beta_{k_1}, \beta_{k_1 + 1}, \hdots,$ $\beta_{k_1 + k_2}, \hdots, \beta_{k_1 + k_2 + \hdots + k_{3t'}}$ satisfying the statements of the theorem, such that $\beta_{k_1 + \hdots + k_{3t'}}$ is claw-free. Let $k = k_1 + \hdots + k_{3t'}$. Since $k_i \leq \Delta^{7t'+4} + 1$, it holds that $k \leq 3t'(\Delta^{7t'+4}+1)$.

Now, by \cite[Theorem~2.12]{wallpaper}, there exists a sequence of separations $\S_{\mathcal{C}}$ and an induced subgraph $\beta_{k+1}$ of $\beta_k$, such that $\beta_{k+1}$ is the central bag for $\S_{\mathcal{C}}$; $C(S)$ is a clique for every $S \in \S$; $\beta_{k+1}$ has no $(w_{k+1}, c)$-balanced separator of size at most $d(1 + \Delta + \hdots + \Delta^t)^{-(k+1)}$; and $\beta_{k+1}$ has no clique cutset. This completes the proof. 
\end{proof}

We use Theorem~\ref{thm:claw-free-bag} to prove several important structural results about connected $S_{t,t,t}$-free graphs with maximum degree $\Delta$ and no bounded balanced separator. We call $\beta_0, \beta_1, \hdots, \beta_{k+1}$ as in Theorem \ref{thm:claw-free-bag} the {\em central bag decomposition} of $G$. This concludes the proof of Step 1 from the outline of the proof in Section \ref{subsec:outline}.

\section{Strip decompositions of the central bag}  \label{sec:centralbag}
Given a graph $G$, a \textit{strip structure} of $G$ is a pair $(H,\eta)$, where $H$ is a graph possibly with loops or parallel edges, $|E(H)|\geq  2$, and $\eta$ is a map with domain the union of $E(H)$ and the set of all pairs $(e,v)$ where $e\in E(H),v\in V(H)$ and $e$ incident with $v$, satisfying the following conditions:
\begin{itemize}
\itemsep0em
\item for every edge $e\in E(H)$, we have $\eta(e)\subseteq V(G)$, and for every $v\in V(H)$ incident with $e$, we have $\eta(e,v)\subseteq \eta(e)$,
\item $\eta(e) \cap \eta(f) = \emptyset$ for all distinct $e, f\in E(H)$,
\item for all distinct $e,f\in E(H)$, $x\in \eta(e)$ and $y\in \eta(f)$ are adjacent in $G$ if and only if $e,f$ share an end-vertex $v$ in $H$, and $x\in \eta(e,v)$ and $y\in \eta(f,v)$.
\end{itemize}
%We call the graph $H$ \emph{the pattern graph} for $G$.
 For $e \in E(H)$ with ends $u, v,$ an {\em $e$-rung of $\eta$} is an
induced path of $\eta(e)$ with vertices $p_1, \hdots , p_k$ in order, where for $1 \leq i \leq k$, $p_i \in \eta(e, u)$ if and only if
$i = 1$, and $p_i \in \eta(e, v)$ if and only if $i = k$.

If in addition
\begin{itemize}
\itemsep0em
\item the sets $\{\eta(e): e\in E(H)\}$ are non-empty and partition $V(G)$, and
\item for every $v\in V(H)$, the union of the sets $\eta(e,v)$ for all $e\in E(H)$ such that $v$ is incident with $e$ is a clique of $G$,
\end{itemize}
then $(H,\eta)$ is called an \textit{elementary strip structure} of $G$. 

Let $G$ be a graph and let $(H,\eta)$ be a strip structure of $G$. Let us extend the domain of $\eta$ by adding to it the union of $V(H)$ and the set of all triangles of $H$, as follows. For each vertex $v\in V(H)$, let $\eta(v)\subseteq V(G)$, and for each triangle $D$ of $H$ let $\eta(D)\subseteq V(G)$, satisfying the following:
\begin{itemize}
\itemsep0em
\item all the sets $\eta(e)\:(e\in E(H))$,$\;\eta(v)\;(v\in V(H))$ and $\eta(D)$ (for all triangles $D$ of $H$) are pairwise disjoint, and their union is $V(G)$,
\item for each $v\in V(H)$, if $x\in \eta(v)$ and $y\in V(G)\setminus \eta(v)$ are adjacent in $G$ then $y\in \eta(e,v)$ for some $e\in E(H)$ incident in $H$ with $v$,
\item for each triangle $D$ of $H$, if $x\in \eta(D)$ and $y\in V(G)\setminus \eta(D)$ are adjacent in $G$ then $y\in \eta(e,u)\cap \eta(e,v)$ for some distinct $u,v\in D$, where $e$ is the edge $uv$ of $H$.
\end{itemize}
Finally, for $e \in E(H)$, let $\Tilde{\eta}(e)$ be the set of all vertices of $\eta(e)$ that do not belong to any $e$-rung of $\eta(e)$. 
%every vertex of $\eta(e)$ that is not in an $e$-rung of $\eta$. 
In this case we say that $(H, \eta)$ is an {\em extended strip decomposition of $G$}. Let $Z \subseteq V(G)$ and $W$ be the set of vertices of $H$ that have degree one in $H$. If 
\begin{itemize}
\itemsep0em
\item $|Z| = |W|$, and for each $z\in Z$ there is a vertex $v\in W$ such that $\eta(e,v) = \{z\}$, where $e$ is the (unique) edge of $H$ incident with $v$,
\end{itemize}
then we say that $(H, \eta)$ is an \emph{extended strip decomposition of $(G,Z)$}.
%, or \emph{$(G,Z)$ admits an extended strip decomposition $(H, \eta)$}. 

If $(A, C, B)$ is a separation of $G$, we call $|C|$ the {\em order} of the separation.
Let $W\subseteq V(H)$. We say that $(H,W)$ is a {\em frame} if
\begin{itemize}
\itemsep0em
\item $H$ is connected,
\item $|W| \ge 3$ and every vertex in $W$ has degree one in $H$,
\item for every separation $(A,C, B)$ of $H$ of order at most two with $W\subseteq B \cup C\ne V(H)$, we have that $|C| = 2$ and $A \cup C$ is a path between the two vertices of $C$.
\end{itemize}
%Observe that for a frame $(H,W)$, $W$ is the set of all vertices of $H$ that have degree one.  A {\em branch} of $(H,W)$ is a path of $H$ with distinct ends, such that both its ends have degree in $H$ different from two, and all its internal vertices have degree two in $H$. Since $(H,W)$ is a frame, it follows that every branch is an induced subgraph of $H$, and every edge of $H$ belongs to a unique branch. If $v\in V(H)$, $\delta_H(v)$ or $\delta(v)$ denotes the set of all edges of $H$ incident with $v$.
%If $(H,W)$ is a frame, we say it is a {\em frame for $(G,Z)$} if $E(H)\subseteq V(G)$, and
%\begin{itemize}
%\itemsep0em
%\item for all distinct $e, f \in E(H)$, $e, f$ have a common end in $H$ if and only if $e, f \in V(G)$ are adjacent in $G$,
%\item $Z$ is the set of edges of $H$ incident with a vertex in $W$.
%\end{itemize}
Let $Z \subseteq V(G)$ and let $(H, \eta)$ be an extended strip decomposition of $(G,Z)$. Let $W$ denote the set of vertices of $H$ of degree one. We say that $(H, \eta)$ is \emph{semi-tame for $(G,Z)$} if
\begin{itemize}
\itemsep0em
\item $H$ has no vertices of degree two,
%\item $(H, \eta)$ is connected, and
\item $(H,W)$ is a frame,
\item $\eta(e) \setminus \Tilde{\eta}(e)$ is non-empty for all $e \in E(H)$, and
\item for every $e \in E(H)$ and every $v \in V(H)$ incident with $e$, it holds that
$\eta(e,v) \cap \tilde{\eta}(e)= \emptyset$.
\end{itemize}

We say that $(H, \eta)$ is {\em tame for $(G, Z)$} if $(H, \eta)$ is semi-tame for $(G, Z)$ and if $\Tilde{\eta}(e)$ is empty for every $e \in E(H)$.

A set $Z$ is \emph{constricted} in a graph $G$ if $G$ does not contain an induced tree $T$ with $Z \subseteq V(T)$. The main result of~\cite{3-in-a-tree} asserts that if $G$ is a connected graph and $Z\subseteq V(G)$ with $|Z| \ge 2$, then $Z$ is constricted in $G$ if and only if for some graph $H$, $(G,Z)$ admits an extended strip decomposition $(H, \eta)$. The following is one of the steps in the proof this result~\cite[Section 5]{3-in-a-tree}:

\begin{theorem}[Chudnovsky, Seymour~\cite{3-in-a-tree}]
\label{thm:extendH}
Let $Z$ be a constricted set in $G'$ and let $G$ be an induced subgraph of $G'$ such that $Z \subseteq V(G)$. Suppose there is a tame extended strip decomposition $(H,\eta)$ of $(G,Z)$.
Then in polynomial time we can obtain a semi-tame extended strip decomposition $(H',\eta')$ of $(G',Z)$ such that a subdivision of $H$ is a subgraph of $H'$.
\end{theorem}

\subsection{Atoms of extended strip decompositions}

Let $G$ be a graph and let $(H, \eta)$ be an extended strip decomposition of $G$. We define the following \emph{atoms} that correspond to $(H,\eta)$ as follows:
\begin{description}
\itemsep0em
\item[vertex atom:] for each vertex $v \in V(H)$, it is $A(v) := \eta(v)$,
\item[edge atom:] for each edge $e = uv \in E(H)$, it is $A(e) := \eta(uv) \setminus (\eta(uv,v) \cup \eta(uv,u))$,
\item[triangle atom:] for each triangle $D = uvw \in T(H)$, it is $A(D) := \eta(uvw)$.
\end{description}
For a vertex $v \in V(H)$, by $\pot(v)$, we denote the set $\bigcup_{uv \in E(H)} \eta(uv,v)$. Each set $\eta(uv,v)$ is called a \emph{segment} of $\pot(v)$. For each atom $A$, let us define its \emph{boundary}, denoted by $\boundary(A)$, as follows:
\begin{itemize}
\setlength{\itemindent}{-.2in}
\itemsep0em
\item If $A$ is a vertex atom associated with a vertex $v$, then $\boundary(A) := \pot(v)$. 
\item If $A$ is an edge atom associated with an edge $uv$, then $\boundary(A) := \pot(u) \cup \pot(v)$. 
\item If $A$ is a triangle atom associated with a triangle $uvw$, then $\boundary(A) := \pot(u) \cup \pot(v) \cup \pot(w)$.
\end{itemize}

%Note that by the definition of an extended strip decomposition, if $v$ is a degree-1 vertex of $H$, then $|\pot(v)|=1$. 
%
%Let $v$ be a vertex of $H$ of degree $d \geq 3$. Observe that $\pot(v)$ has a spanning complete $d$-partite subgraph, and so if $G$ is a graph with maximum degree at most $\Delta$, then $|\pot(v)| \leq d/(d-1) \cdot \Delta$. Summing up, we obtain that $|\pot(v)| \leq 3/2 \cdot \Delta$ for every $v \in V(H)$.
%
%A vertex of degree at least three in $H$ is called a \emph{branch vertex}. 
%
%Let $F = (H,W)$ be a frame. Then, each branch vertex $v$ of $H$ corresponds to a clique $K_F(v)$ of $F$, we call it a \emph{branch clique}. Furthermore, the vertices of $K_F(v)$ are in one-to-one correspondence to segments of $\pot(v)$. 

\subsection{A tame extended strip decomposition with small atoms}

A breakthrough in the understanding of the structure of claw-free graphs is due to the seminal work of Chudnovsky and Seymour, where a decomposition theorem for claw-free graphs was given in a series of papers. In particular, the following can be derived from \cite[Theorem 7.2]{Claw5}.
\begin{theorem}[Chudnovsky, Seymour~\cite{Claw5}]
\label{thm:clawstructure2}
Let $G$ be a connected claw-free graph. Then, one of the following holds.
\begin{itemize}
\itemsep0em
\item We have $\alpha(G)\leq 3$.
\item $G$ is a fuzzy long circular interval graph.
\item $G$ admits an elementary strip structure $(H,\eta)$ such that for every $e\in E(H)$, either $\alpha(G[\eta(e)]) \leq 4$ or $G[\eta(e)]$ is a fuzzy long circular interval graph. Moreover, we can find this elementary strip structure in polynomial time.
\end{itemize}
\end{theorem}

%Reference:
%Danny Hermelin, Matthias Mnich, Erik Jan van Leeuwen, and Gerhard Woeginger. Domination When the Stars Are Out. 
%ACM Transactions on Algorithms, Volume 15, Issue 2, May 2019, pp 1–90.

The complexity in the last item follows from Theorem 6.8 in \cite{domination}. We do not give the precise definition of a ``fuzzy long circular interval graph'' here since all we need is the fact that fuzzy long circular interval graphs with maximum degree at most $\Delta$ have treewidth bounded by a constant multiple of $\Delta$ (see \cite{wallpaper} for a precise definition):

\begin{theorem}[Abrishami et al.~\cite{wallpaper}]
\label{thm:circtw}
If $G$ is a fuzzy long circular interval graph with maximum degree at most $\Delta$, then $\tw(G) \leq 4\Delta+3$.
\end{theorem}

The following is an easy observation.

\begin{observation}
\label{obs:bdd_alpha_bdd_tw}
If $G$ is a graph with maximum degree $\Delta$, then $\tw(G) \leq |V(G)| \leq \alpha(G) (\Delta +1)$.
\end{observation}

Theorem \ref{thm:clawstructure2}, Theorem \ref{thm:circtw}, and Observation \ref{obs:bdd_alpha_bdd_tw}, together imply the following corollary.

\begin{corollary}
\label{cor:claw_free_tw}
If $G$ is a connected claw-free graph with maximum degree at most $\Delta$, then either $\tw(G) \leq 4\Delta+3$, or $G$ admits an elementary strip structure $(H,\eta)$ such that for every $e\in E(H)$, we have $\tw(G[\eta(e)]) \leq 4\Delta+4$.
\end{corollary}

%Let $G$ be a graph and let $w: V(G) \to [0, 1]$ be a weight function defined on the vertices of $G$. For $X \subseteq V(G)$, let $w(X) = \sum_{x \in X} w(x)$, and denote $w(V(G))$ by $w(G)$. Let $c \in [\frac{1}{2}, 1)$. A set $X \subseteq V(G)$ is a {\em $(w, c)$-balanced separator} if every connected component $D$ of $G \setminus X$ satisfies $w(D) \leq c$. A $(w, c)$-balanced separator $X$ is {\em $d$-bounded} if $|X| \leq d$. We recall the following result (probably already in Tara's part).

%\begin{lemma}
%Let $G$ be a graph. If $\tw(G) \leq k$, then for every weight function $w$ and for every $c \in [\frac{1}{2},1)$, $G$ has a $(k+1)$-bounded $(w, c)$-balanced separator.
%\end{lemma}

%\paragraph{What we get from the separations part.} We are given an $n$-vertex $S_{t,t,t}$-free connected graph $G$ with maximum degree at most $\Delta$. Let $\beta$ be the central bag that we get after decomposing by $s$-claw-centered cutsets and clique cutsets. Then, $\beta$ is a connected claw-free graph with maximum degree at most $\Delta$ and with no clique cutset. Furthermore, we may assume that $\beta$ has no bounded balanced separator.\\

Let $G$ be a connected $n$-vertex $S_{t, t, t}$-free graph with maximum degree $\Delta$, and assume $G$ has no $(w, c)$-balanced separator of size at most $d$ for some $w: V(G) \to [0, 1]$ with $w(G) = 1$, $c \in [\frac{1}{2}, 1)$, and $d > (3t+1)\Delta(1 + \Delta + \hdots + \Delta^t)^{3t(\Delta^{7t+4}+1)}$. Let $\beta_0, \hdots, \beta_{k+1}$ be the central bag decomposition for $G$. For the remainder of this section, let $\beta = \beta_{k+1}$. By Theorem \ref{thm:claw-free-bag}, $\beta$ is connected, claw-free, and has no clique cutset. Since $G$ has no $(w, c)$-balanced separator of size at most $d$, it holds that $\tw(G) > 4\Delta + 3$, and so it follows from Corollary \ref{cor:claw_free_tw} that $\beta$ admits an elementary strip structure $(H', \eta')$ such that for every $e\in E(H')$, we have $\tw(\beta[\eta'(e)]) \leq 4\Delta+4$. Recall that for an edge $e = uv \in E(H')$, the edge atom $A(e)$ is defined as $A(e) = \eta'(e) \setminus (\eta'(e,v) \cup \eta'(e,u))$. Since $A(e) \subseteq \eta'(e)$, it follows that for every $e\in E(H')$, the graph $\beta[A(e)]$ has treewidth at most $4 \Delta + 4$.

%Since $\eta'(e,v)$ and $\eta'(e,u)$ are cliques, we have $|\eta'(e,v) \cup \eta'(e,u)| \leq 2 \Delta +2$, and therefore for every $e\in E(H')$, the edge atom $A(e)$ has treewidth bounded by a constant multiple of $\Delta$.\\

%Recall that for every $v \in V(H')$, we have $|\pot(v)| \leq 3/2 \cdot \Delta$.

\begin{lemma}
\label{lem:small_atoms}
Let $(H', \eta')$ be an elementary strip structure of $\beta$. Let $w$ be a weight function on $\beta$ such that $w(\beta) = 1$. For every $\sigma \in (0,1)$, if there exist an edge $e \in E(H')$ with $w(A(e)) \geq \sigma$, then there is a $(w, c)$-balanced separator of size $6\Delta + 7$ in $\beta$, where $c = \max\{\frac{1}{2}, 1-\sigma\}$.
\end{lemma}

%\begin{lemma}
%Let $(H', \eta')$ be an elementary strip structure of $\beta$. Let $w$ be a weight function on $\beta$ such that $w(\beta) = 1$. For every $\sigma \in (0,1)$, if there exist an edge $e \in E(H')$ with $w(A(e)) \geq \sigma$, then for some $c \in [\frac{1}{2},1)$, there is a $(6\Delta+7)$-bounded $(w, c)$-balanced separator in $\beta$.
%\label{lem:small_atoms}
%\end{lemma}

\begin{proof}
Let $\sigma \in (0,1)$ and assume that there exist an edge $e = uv \in E(H')$ such that $w(A(e)) = \sigma^*$ for some $\sigma^* \geq \sigma$. Let us define a weight function $w_A$ on $A(e)$ as follows: for $v \in A(e)$, we set $w_A(v) = \frac{w(v)}{\sigma^*}$ (and so $w_A(A(e)) = 1$). Since $\tw(\beta[A(e)]) \leq 4 \Delta + 4$, by Lemma \ref{lem:tw_to_bsp}, there exists a $(w_A, \frac{1}{2})$-balanced separator of $A(e)$ of size $4\Delta + 5$. Let us call this separator $X$. Note that since for every $v \in V(H')$, the union of the sets $\eta'(e, v)$ for all $e \in E(H')$ with $v \in e$ is a clique of $\beta$, we have $|\pot(v)| \leq \Delta + 1$ for every $v \in V(H')$. We claim that $X' = X \cup \pot(u) \cup \pot(v)$ is a $(w, c)$-balanced separator of size $(6\Delta+7)$ in $\beta$, where $c = \max\{\frac{1}{2}, 1-\sigma\}$. We have $|X'| \leq |X| + |\pot(u)| + |\pot(v)| \leq 6\Delta+7$. Let $D$ be a component of $\beta \setminus X'$. Note that either $D \cap A(e) = \emptyset$ or $D \subseteq A(e)$. If $D \cap A(e) = \emptyset$, then since $w(A(e)) \geq \sigma$, we have $w(D) \leq 1-\sigma$. If $D \subseteq A(e)$, then since $D$ is also a component of $A(e) \setminus X$, we have $w_A(D) \leq \frac{1}{2}$, and therefore $w(D) \leq \frac{1}{2} \cdot \sigma ^* \leq \frac{1}{2}$. Thus, every component $D$ of $\beta \setminus X'$ satisfies $w(D) \leq c$, where $c = \max\{\frac{1}{2}, 1-\sigma\}$. Hence, $X'$ is a $(w, c)$-balanced separator of size $(6 \Delta + 7)$ in $\beta$.
\end{proof}

\begin{lemma}
\label{lemma:H_3_connected}
For every $\sigma \in (0,1)$, there exists an elementary strip structure $(H, \eta)$ of $\beta$ where $H$ is 3-connected, and for a weight function $w$ on $\beta$, we have $w(A(e)) < \sigma$ for every edge $e \in E(H)$.
\end{lemma}

\begin{proof}
Let $(H', \eta')$ be an elementary strip structure of $\beta$.
By Theorem \ref{thm:claw-free-bag}, for every $c \in [\frac{1}{2}, 1)$ and for some $d > 6\Delta + 7$, it holds that $\beta$ has no $(w_\S, c)$-balanced separator of size at most $d$. Therefore,  we may assume that $w(A(e)) < \sigma$ for every edge $e \in E(H')$ and for every $\sigma \in (0,1)$. Suppose $\{v\}$ is a cutset of size one in $H'$. Then, $\pot(v)$ is a clique cutset in $\beta$. Since $\beta$ has no clique cutset, it follows that $H'$ is 2-connected. 
%Moreover, since $\beta$ has no clique cutset, the graph $H'$ is 2-connected. Indeed, if $\{v\}$ is a cutset of size one in $H'$, then $\pot(v)$ is a clique cutset in $\beta$. Our next goal is to modify $(H', \eta')$ to obtain the following.

Let $\{a, b\} \subseteq V(H')$ be a cutset in $H'$. Then, $\pot(a) \cup \pot(b)$ is a cutset in $\beta$. Let $c = \max\{\frac{1}{2}, 1-\sigma\}$ and let $w$ be a weight function on $\beta$. Since $|\pot(a) \cup \pot(b)| \leq 2 \Delta + 2$ and since $\beta$ has no bounded $(w,c)$-balanced separator, one of the components of $\beta \setminus (\pot(a) \cup \pot(b))$ has weight greater than $c$. Let $T$ be the vertex set of this component and let $S$ be the vertex set of the union of the other components of $\beta \setminus (\pot(a) \cup \pot(b))$. Then, if $ab \in E(H')$, then $T \neq \eta'(ab) \setminus (\eta'(ab, a) \cup \eta'(ab, b))$, and therefore $\eta'(ab) \setminus (\eta'(ab, a) \cup \eta'(ab, b))$ is contained in $S$. Let $H$ be the graph obtained from $H'$ by deleting the edges $e \in E(H')$ such that $\eta'(e) \cap S \neq \emptyset$ and the vertices $v \in V(H')$ such that $\pot(v) \cap S \neq \emptyset$,
%$\eta'(e, v) \cap S \neq \emptyset$ for some $e \in E(H')$,
and by adding a new edge $f$ with ends $a, b$. Let $\eta$ be the map obtained from $\eta'$ restricted to $T \cup \pot(a) \cup \pot(b)$ by setting 
\begin{itemize}
\itemsep0em
\item $\eta(f) \setminus (\eta(f, a) \cup \eta(f, b)) =S$,
\item $\eta(f, a) = \bigcup_{u \in S} \eta'(au, a)$, and
\item $\eta(f, b) = \bigcup_{u \in S} \eta'(bu, b)$.
\end{itemize}
Repeating this procedure until the pattern graph has no cutset of size 2 yields an elementary strip structure $(H, \eta)$ of $\beta$ where $H$ is 3-connected. Observe also that the atoms generated by this operation preserves having the property of ``small'' weight since $S$ is the union of the ``small'' weight components. Hence, we have $w(A(e)) < \sigma$ for every edge $e \in E(H)$.
\end{proof}

Let $m_1, m_2, m_3$ be three distinct branch vertices of $H$, and
%that form a stable set. 
for $i=1,2,3$, let $M_i \subseteq V(\beta)$ be the branch clique that corresponds to $m_i$. 
%Note that the sets $M_1, M_2, M_3$ are pairwise disjoint. 
Let $G'$ be the graph obtained from $G$ by adding three new vertices $v_1, v_2, v_3$ with $N_{G'}(v_i) = M_i$. Since $G$ has bounded degree, it holds that $G'$ also has bounded degree. Let $\beta'' = \beta \cup \{v_1, v_2, v_3\}$ and let
$$Y = \{v: v \in V(G') \setminus \beta'' \text{ s.t. dist$_{G'}(v, \{v_1, v_2, v_3\}) \leq t+1$} \}.$$ 
Observe that $|Y| \leq 3 \sum_{i=1}^{t+1} \Delta^i$. Let $G'' = G' \setminus Y$. \

\bigskip

\begin{lemma}
\label{lem:v1_v2_v3_constricted}
The set $\{v_1, v_2, v_3\}$ is constricted in $G''$.
\end{lemma}

\begin{proof}
Suppose  $T$ is an induced tree in $G''$ containing $v_1, v_2, v_3$, and with $V(T)$ minimal. Then, either $T$ is a path with both end-vertices in $\{v_1, v_2, v_3\}$, or $T$ is a subdivided claw and $v_1, v_2, v_3$ all have degree one in $T$. Since $v_1, v_2, v_3$ are simplicial vertices of $G''$, it follows that $T$ is a subdivided claw. Let $a$ be the unique vertex of degree three in $T$, and let $T$ consist of paths $T_1, T_2, T_3$ where, for $i=1,2,3$, $T_i$ is a path from $a$ to $v_i$. Let $b_i$ be the neighbor of $a$ in $T_i$. Since $\beta''$ is claw-free, we have $\{a, b_1, b_2, b_3\} \setminus \beta'' \neq \emptyset$. Let $u \in \{a, b_1, b_2, b_3\} \setminus \beta''$. Since $V(G'') \cap Y = \emptyset$, it follows that dist$_{G'}(u, v_i) \geq t+2$ for $i=1,2,3$. But then $T \setminus \{v_1, v_2, v_3\}$ contains an induced $S_{t,t,t}$, and so $G \setminus Y$ contains an induced $S_{t,t,t}$, a contradiction.
\end{proof}

Note that $\beta''$ is an induced subgraph of $G''$ and it is claw-free. Let $H''$ be the graph obtained from $H$ by adding three new vertices $l_1, l_2, l_3$ where for $i=1,2,3$, the unique neighbor of $l_i$ in $H''$ is $m_i$. Let $\eta''$ be the map obtained from $\eta$ by additionally setting $$\eta''(l_i m_i) = \eta'' (l_im_i, l_i) = \eta''(l_im_i, m_i) = \{v_i\}$$ for $i=1,2,3$. It is now straightforward to see that $(H'', \eta'')$ is a semi-tame extended strip decomposition of $(\beta'', \{v_1, v_2, v_3\})$.

We summarize the results of this section in the following theorem.

\begin{theorem}
\label{thm:esd-summary}
Let $G$ be an $n$-vertex $S_{t,t,t}$-free connected graph with maximum degree at most $\Delta$ and no $(w, c)$-balanced separator of size $d$, where $w(v) = \frac{1}{n}$ for all $v \in V(G)$. Let $\beta_0, \hdots, \beta_{k+1}$ be the central bag decomposition of $G$, and let $w_{k+1}$ be the weight function on $\beta_{k+1}$.
\begin{enumerate}
\itemsep0em
\item For every $\sigma \in (0,1)$, there exists an elementary strip structure $(H, \eta)$ of $\beta_{k+1}$ where $H$ is 3-connected and $w_{k+1}(A(e)) < \sigma$ for every edge $e \in E(H)$.
\item Let $m_1, m_2, m_3$ be three distinct branch vertices of $H$, and for $i=1,2,3$, let $M_i \subseteq V(\beta)$ be the branch clique that corresponds to $m_i$. Let $G'$ be the graph obtained from $G$ by adding three new vertices $v_1, v_2, v_3$ with $N_{G'}(v_i) = M_i$. Let $\beta'' = \beta_{k+1} \cup \{v_1, v_2, v_3\}$ and let
\[
Y = \{v: v \in V(G') \setminus \beta'' \text{ s.t. } \dist_{G'}(v, \{v_1, v_2, v_3\}) \leq t+1 \}.
\]
Let $G'' = G' \setminus Y$. The set $\{v_1, v_2, v_3\}$ is constricted in $G''$, and in polynomial time, we can find a semi-tame extended strip decomposition $(H'', \eta'')$ of $(\beta'', \{v_1, v_2, v_3\})$ such that $w_{k+1}(A(e)) < \sigma$ for every edge $e \in E(H'')$.
\end{enumerate}
\end{theorem}

\section{Extending extended strip decompositions} \label{sec:extending}

In this section, we prove an essential inductive lemma stating that we can grow an extended strip decomposition of a central bag of $G$ to an extended strip decomposition of $G$, while preserving a bound on the weight of each atom of the strip decompositions. Theorem~\ref{thm:esd-summary} and this lemma are the key components of the proof of Theorem~\ref{thm:combinatorial}. 

\begin{lemma}
\label{lemma:esd-induction}
Fix integers $\Delta \geq 3$, $t, \rho \geq 1$, and $d \geq (9/2\cdot \Delta^2)^{\Delta(2t+2)} + (1 + (3t+1)\Delta 2^{\Delta})\rho$, and real number $c \in [\frac{1}{2}, 1)$. Let $G$ be a connected graph with maximum degree $\Delta$ and let $w'$ be a weight function on $V(G)$ with $w'(G) = 1$. Assume $G$ has no $(w', c)$-balanced separator of size at most $d$. Let $\beta \subseteq G$ be a connected set and let $Z \subseteq \beta$ be of size 3. Let $\S$ be an $A$-loosely laminar sequence of separations of $G$ satisfying the following conditions:
\begin{itemize} 
\itemsep0em
\item For every vertex $v \in V(G)$, the set $\{S \in \S \mid v \in C(S)\}$ has size at most $2^{\Delta}$; 
\item $|C(S)| \leq (3t+1)\Delta$ for every $S \in \S$; and 

\item The diameter of $C(S)$ is at most $2t+2$ for every $S \in \S$. 
\end{itemize} 
Let $\beta'$ be the central bag for $\S$, let $w$ be the weight function on $\beta'$, and assume that $\beta \subseteq \beta'$. Let $Y \subseteq G \setminus \beta$ be a set of size at most $\rho$ such that $Z$ is constricted in $G - Y$. Let $M'$ (resp. $M$) be the connected component of $G - Y$ (resp. $\beta' - Y)$ containing $\beta$. Finally, assume that for each component $D$ of $G - Y$, except for $M'$, it holds that $w(D) < 1-c$. 

Suppose we are given a semi-tame extended strip decomposition $(H, \eta)$ of $(M, Z)$, whose every atom is of weight at most $1-c$ (under $w$). Then in polynomial time we can compute a semi-tame extended strip decomposition $(H', \eta')$ of $(M', Z)$, where $H$ is a subgraph of $H'$ and every atom is of weight at most $1-c$ (under $w')$. 
\end{lemma}

\begin{proof}
Recall that every vertex $z \in Z$ belongs to $\eta(e,v)$ for some $e \in E(H)$ and $v \in e$,
and $\tilde{\eta}(e) \cap \eta(e,v) = \emptyset$.
Consequently, $Z \cap \bigcup_{e \in E(H)}\tilde{\eta}(e) = \emptyset$. Thus $(H,\eta)$ can be seen as a tame extended strip decomposition of $(M - \bigcup_{e \in E(H)}\tilde{\eta}(e), Z)$ (here we slightly abuse the definition -- formally, the function $\eta$ needs to modified by removing the vertices from $\bigcup_{e \in E(H)}\tilde{\eta}(e)$ from all sets).
Furthermore, $(M - \bigcup_{e \in E(H)}\tilde{\eta}(e))$ is an induced subgraph of $M'$.
Thus we can apply Theorem~\ref{thm:extendH} to obtain a semi-tame extended strip decomposition $(H',\eta')$ of $M'$, such that a subdivision of $H$ is a subgraph of $H'$.
In what follows, the sets $\pot( \cdot )$ and $\boundary( \cdot )$ are defined with respect to $(M',\eta')$.
For a vertex $x$ and a set $X$, whenever we write $N(x)$ or $N(X)$ without any subscript, we mean the neighborhood in $G$.

Let us fix a frame $F'$ for $(M',\eta')$, i.e., we choose one rung for each strip.
For each branch vertex $v$ of $H'$ there is a branch clique $K_{F'}(v)$ in $F'$.
In particular, each segment of $\pot(v)$ is represented by a single vertex of $K_{F'}(v)$.
Furthermore, as each branch vertex of $H$ is a branch vertex of $H'$, we observe that for each branch vertex $v$ of $H$ there is a corresponding branch clique in $F'$.

We say that a vertex $x$ of $M'$ is \emph{heavy} if there is a branch clique $K$ in $F'$, such that $x$ has at most one non-neighbor in $K$.
Let $\Heavy$ be the set of all heavy vertices in $M'$.

First, let us observe that $\Heavy$ contains all boundaries of atoms of $(H,\eta)$ and of $(H',\eta')$.

\begin{claim}\label{clm:heavyH}
For every $x \in V(M)$, if $x \in \bigcup_{uv \in E(H)} \eta(uv,v)$, then $x \in \Heavy$.
Consequently, every connected subset of $V(M) \setminus \Heavy$ is fully contained in a single atom of $(H,\eta)$.
\end{claim}
\begin{claimproof}
Let $uv \in E(H)$ be such that $x \in \eta(uv,v)$ and let $K$ be the branch clique of $F'$ that corresponds to $v$.
Thus $x$ belongs to a segment $\eta(uv,v)$ of $\bigcup_{u'v \in E(H)} \eta(u'v,v)$.
Recall that $K$ contains exactly one vertex from each segment of $\bigcup_{u'v \in E(H)} \eta(u'v,v)$.
Thus $x$ is adjacent to all vertices of $K$, possibly except for the one chosen for the segment $\eta(uv,v)$. 

The second statement follows from the fact that every path in $M$ with endpoints in distinct atoms of $(H,\eta)$ must cross their boundaries which, as already observed, are contained in $\Heavy$.
\end{claimproof}

\begin{claim}\label{clm:heavyHprime}
For every $x \in V(M')$, if $x \in \bigcup_{uv \in E(H')} \eta'(uv,v)$, then $x \in \Heavy$.
Consequently, every connected subset of $V(M') \setminus \Heavy$ is fully contained in a single atom of $(H',\eta')$.
\end{claim}
\begin{claimproof}
The proof is similar (but simpler) as the one of Claim~\ref{clm:heavyH}:
If $x \in \eta'(uv,v)$, then $x$ is adjacent to every vertex from $K_{F'}(v)$, possibly except for the vertex chosen for the segment $\eta'(uv,v)$.
\end{claimproof}

Let $\cD$ be the set of connected components of $M - \Heavy$.

\begin{claim}\label{clm:weightD}
For each $D \in \cD$ it holds that $w(D) \leq 1-c$.
\end{claim}
\begin{claimproof}
By Claim~\ref{clm:heavyH} each $D \in \cD$ is contained in some atom of $(H,\eta)$, and the weight of each atom is at most $1-c$.
\end{claimproof}

A separator $S \in \cS$ is \emph{marginal} if $C(S) \cap Y \neq \emptyset$.
By $\cS_{m}$ we denote the family or marginal separators from $\cS$.

\begin{claim}\label{clm:Yprimesmall}
$|\bigcup_{S \in \cS_m} C(S)| \leq (3t+1)\Delta 2^{\Delta}|Y|$. 
\end{claim}
\begin{claimproof}
Since for every $v \in Y$, we have $|\{S \in \S \mid v \in C(S)\}| \leq 2^{\Delta}$, we conclude that $|\cS_m| \leq 2^{\Delta}|Y|$.
The claim follows from the fact that for all $S \in \cS$ we have $|C(S)| \leq (3t+1)\Delta$.
\end{claimproof}

Let $Y'$ denote $\bigcup_{S \in \cS_m} (A(S) \cup C(S))$.
Define $\cS' := \cS \setminus \cS_m$.
For $S \in \cS'$, let $A'(S)$ denote $A(S) \cap M'$.
We say that a separation $S \in \cS'$ is \emph{serious} if $C(S)$ is not contained in a single connected component of $M - \Heavy$.

\begin{claim}\label{clm:seriousheavy}
For every serious separation $S$, it holds that $C(S)$ intersects $\Heavy$.
Furthermore, for any $x \in C(S) \cap \Heavy$ it holds that $C(S) \subseteq N^{2t+2}[x]$.
\end{claim}
\begin{claimproof}
The first statement follows immediately since $C(S)$ is connected.
The second statement follows by observing that the diameter of $C(S)$ is at most $2t+2$.
\end{claimproof}

Define $R := \Heavy \cup \{C(S)  ~|~ S  \in \cS' \text{ is serious} \}$.
Let us point out that the set $R$ might be arbitrarily large.
However, we will show that its interaction with the rest of the graph is simple.

\begin{claim} \label{clm:attachA}
For each $S \in \cS'$, there is at most one $D \in \cD$ such that $A'(S)- R$ attaches to $D - R$.
\end{claim}
\begin{claimproof}
For contradiction, suppose that for some $S \in \cS$, the set $A'(S) -R$ attaches to $D -R$ and $D' -R$ for distinct $D,D' \in \cD$.
Recall that $N(A'(S)) \subseteq C(S) \cup Y$, so $N_{G-Y}(A'(S)) \subseteq C(S) \setminus Y = C(S)$, as $S$ is not marginal.
Consequently, $C(S)$ intersects $D$ and $D'$, and thus $S$ is a serious separation.
In particular, we have $C(S) \subseteq R$.
Therefore there are no edges from $A'(S)-R$ to $D - R$ and to $D' - R$, a contradiction.
\end{claimproof}

We introduce a classification of components $D'$ of $M' - R - Y'$ with respect to their relation to the elements of $\cD$ and to $A'(S)$ for $S \in \cS'$:
\begin{enumerate}[{Type~}1:]
\item $D'$ is entirely contained in $A'(S)$ for some $S \in \cS'$.
%Note that in this case $S$ is a serious separation as $A'(S)$ is connected\prz{why connected?} and $N_{G-Y}(A'(S)) \subseteq C(S)$.
\item $D'$ is entirely contained in some $D \in \cD$.
\item $D'$ is not of type~1 nor~2. By Claim~\ref{clm:attachA} then there is some $D \in \cD$
such that $D' \subseteq D \cup \bigcup_{S \in \cS' ~:~ C(S) \subseteq D} A'(S)$.
Note that in this case, if $D' \cap A'(S) \neq \emptyset$, then $S$ is not serious.
\end{enumerate}

\begin{claim}\label{clm:componentsaresmall}
For each component $D'$ of $M' - R - Y'$ it holds that $w'(D) \leq 1-c$.
\end{claim}
\begin{claimproof}
If $D'$ is of type~1, then $D' \subseteq A'(S) \subseteq A(S)$ for some $S \in \cS'$ and thus $w'(D) \leq w'(A(S)) \leq 1-c$ by the assumption on $\cS$.
If $D'$ is of type~2, then $D' \subseteq D$ for some $D \in \cD$.
As for every $v \in D$ it holds that $w'(v) \leq w(v)$, we obtain $w'(D') \leq w'(D) \leq w(D) \leq 1-c$ by Claim~\ref{clm:weightD}.
So assume that $D'$ is of type~3, i.e., $D' \subseteq D \cup \bigcup_{S \in \cS' ~:~ C(S) \subseteq D} A'(S)$ for some $D \in \cD$.
%Recall that $w$ was obtained from $w'$ by transferring the weight of each $A(C)$ to $\cent(C) \in C \subseteq D$.
For every $S$ it holds that $w'(C(S) \cup A(S)) \leq w(C(S))$ and consequently 
\[
w'(D') \leq w' \left(D \cup \bigcup_{S \in \cS' ~:~ C(S) \subseteq D} A'(S) \right) \leq w' \left(D \cup \bigcup_{S \in \cS' ~:~ C(S) \subseteq D} A(S) \right) \leq w(D) \leq 1-c\]
by Claim~\ref{clm:weightD}.
\end{claimproof}

Now we aim to analyze how elements of $\cD$ interact with $R \setminus Y'$.

\begin{claim} \label{clm:attachDprime}
For each $D \in \cD$ it holds that $D \cap (R \setminus Y') \subseteq N^{2t+2}[N[D] \cap \Heavy]$.
\end{claim}
\begin{claimproof}
As $D \cap \Heavy$ is clearly contained in $N^{2t+2}[N[D] \cap \Heavy]$, let us consider $((R \setminus \Heavy) \setminus Y') \cap D$.
Let $\cS'_D \subseteq \cS'$ be the set of serious separations that intersect $D$.

Consider $S \in \cS'_D$. By Claim~\ref{clm:seriousheavy} we know that $C(S) \cap \Heavy \neq \emptyset$.
Furthermore, as $C(S)$ is connected and intersects $D$, we observe that $C(S) \cap \Heavy \cap N(D) \neq \emptyset$.
Thus, again using Claim~\ref{clm:seriousheavy} for any $x \in C(S) \cap \Heavy \cap N(D)$, we observe that 
$C(S) \subseteq N^{2t+2}[x] \subseteq  N^{2t+2}[N[D] \cap \Heavy]$.

In summary, we obtain
\[
(R \setminus \Heavy \setminus Y') \cap D \subseteq \bigcup_{S \in \cS'_D} C(S)  \subseteq N^{2t+2}[N[D] \cap \Heavy],
\]
which completes the proof of the claim.
\end{claimproof}

Now let us focus on atoms of $(H',\eta')$.

\begin{claim} \label{clm:attachQ}
For each atom $Q$ of $(H',\eta')$, it holds that $|N[Q] \cap (R \setminus Y')| \leq (9/2 \cdot \Delta^2)^{\Delta(2t+2)}$.
\end{claim}
\begin{claimproof}
For an atom $Q$, let $\cD_Q$ be the set of those $D \in \cD$ for which $Q \cap D \neq \emptyset$.
Note that by Claim~\ref{clm:heavyHprime} each $D \in \cD_Q$ is fully contained in $Q$.

First, we focus on $N[Q] \cap \Heavy$.
Recall that every heavy vertex from $N[Q]$ either belongs to the boundary of $Q$, or is adjacent to a vertex of the boundary of $Q$.
Thus $N[Q] \cap \Heavy \subseteq N[\boundary(Q)]$.
Since $|\boundary(Q)| \leq 9/2 \cdot \Delta$, we obtain that $|N[Q] \cap \Heavy| \leq 9/2 \cdot \Delta^2$.

As the boundary of $Q$ is contained in $\Heavy$ by Claim~\ref{clm:heavyHprime},
we have that $N[Q] \cap ((R \setminus \Heavy) \setminus Y') = Q \cap ((R \setminus \Heavy) \setminus Y')$.
We observe that $Q \cap ((R \setminus \Heavy) \setminus Y') \subseteq \bigcup_{D \in \cD_Q} (D \cap (R \setminus Y'))$.

Consider $D \in \cD_Q$. By Claim~\ref{clm:attachDprime} we observe that
\[
D \cap (R \setminus Y') \subseteq N^{2t+2}[N[D] \cap \Heavy]  \subseteq N^{2t+2}[N[Q] \cap \Heavy],
\]
and thus
$Q \cap ((R \setminus \Heavy) \setminus Y') \subseteq N^{2t+2}[N[Q] \cap \Heavy]$.
Consequently we obtain that $N[Q] \cap (R \setminus Y') \subseteq N^{2t+2}[N[Q] \cap \Heavy]$ and thus 
$|N[Q] \cap (R \setminus Y')| \leq (9/2 \cdot \Delta^2)^{\Delta(2t+2)}$.
\end{claimproof}

To reach the final contradiction suppose that there is an atom $Q$ of $(H',\eta')$ such that $w'(Q) > 1-c$.
Let $X = (N[Q] \cap (R \setminus Y')) \cup Y \cup \bigcup_{C \in \cS_m} C(S)$ and consider the components of $G - X$.
Note that each component $D$ of $G -X$ is of one of the following types:
\begin{enumerate}
\itemsep0em
\item is contained in a component of $G - Y$ other than $M'$; in this case we have $w'(D) < 1-c$ by the assumption of the lemma,
\item is contained in $A(S)$ for some $S \in \cS_m$; in this case we have $w'(D) < 1-c$ by the assumption on $\cS$,
\item is contained in $Q$; in this case we have $w'(D) < 1-c$ by Claim~\ref{clm:componentsaresmall},
\item is contained in $M' \setminus Q$; in this case we have $w'(D) < c$ as $w'(Q)  > 1-c$.
\end{enumerate}

As, by Claim~\ref{clm:Yprimesmall} and Claim~\ref{clm:attachQ}, we have
\begin{align*}
|X| \leq & \ |(N[Q] \cap (R \setminus Y'))| +  |Y| +  \left|\bigcup_{C \in \cS_m} C\right| \\
 \leq & \ (9/2 \cdot \Delta^2)^{\Delta(2t+2)} +  |Y| +  (3t+1)\Delta 2^{\Delta}|Y|\\
 = & \ (9/2 \cdot \Delta^2)^{\Delta(2t+2)} + (1 + (3t+1)\Delta 2^{\Delta}) |Y| \\
 = & \ d,
\end{align*}
we conclude that $G$ has a $(w', c)$-balanced separator $X$ of size at most $d$, which contradicts our assumption of the lemma.
This completes the proof.
\end{proof}

\section{Proofs of main theorems}

In this section, we first prove our main combinatorial result, Theorem~\ref{thm:combinatorial}, and then we use Theorem~\ref{thm:combinatorial} to prove our main algorithmic result, Theorem~\ref{thm:algorithm}.

\subsection{Combinatorial result} \label{sec:combinatorial}

We begin with a lemma about central bags. 
\begin{lemma}
\label{lemma:helper-combinatorial}
Let $c \in [\frac{1}{2}, 1)$ and let $d, \Delta, t$ be positive integers with $d \geq (1 + \Delta + \hdots + \Delta^t)^{3t\Delta^{7t+4}}$. Let $G$ be a graph with maximum degree $\Delta$ and no $(w, c)$-balanced separator of size at most $d$. Let $\S$ be an $A$-loosely laminar sequence of separations such that the diameter of $C(S)$ is at most $2(t+1)$ for every $S \in \S$, let $\beta$ be the central bag for $\S$, and let $w_\S$ be the weight function on $\beta$. Let $Q, Y \subseteq \beta$, and let $C$ be a connected component of $\beta \setminus Y$. Let $Q' = N^{2(t+1)\Delta}[Q] \setminus C$. Let $D$ be a component of $G \setminus Q'$ such that $C \cap D = \emptyset$. Then,
one of the following holds: 
\begin{enumerate}[(i)]
\itemsep0em
    \item There exists $S \in \S$ such that $D \subseteq A(S)$, and so $w(D) \leq 1-c$, or 
    
    \item $D$ meets a unique component $C'$ of $\beta \setminus Q$, and
$w(D) \leq w_\S(C')$.
\end{enumerate}
\end{lemma}
\begin{proof}
Let $S \in \S$ and suppose $D \cap A(S) \neq \emptyset$. If $C(S) \cap Q \neq \emptyset$, then $C(S) \subseteq Q' \cup C$ and $D \subseteq A(S)$, so (i) holds. Therefore, we may assume that $C(S) \cap Q = \emptyset$ for all $S \in \S$ such that $D \cap A(S) \neq \emptyset$. Then, there exists a component $C'$ of $\beta \setminus Q$ such that $C(S) \subseteq C'$ for all $S \in \S$ such that $D \cap A(S) \neq \emptyset$. Now, it follows that $C'$ is the unique connected component of $\beta \setminus Q$ such that $D \cap C' \neq \emptyset$. By (iv) of Lemma~\ref{lemma:central_bags},  $w_\S(C') \geq w(C') + \sum_{S \in \S, C(S) \subseteq C'} w(A(S))$, so (ii) holds. 
\end{proof}

Now, we are ready to prove the main combinatorial result of the paper. 

\thmcomb*
\begin{proof}
Let $c = 1 - \frac{1}{10\Delta}$ and let $d = (1 + \Delta + \hdots + \Delta^t)^{3t\Delta^{7t+4}}$. Let $G$ be an $n$-vertex $S_{t,t,t}$-free graph with maximum degree $\Delta$, and assume $G$ has no $c$-balanced separator of size at most $d$. Let $w: V(G) \to [0, 1]$ be the weight function on $V(G)$ such that $w(v) = \frac{1}{n}$ for every $v \in V(G)$. Let $k = 3t\Delta^{7t+4}$, and let $\beta_0, \beta_1, \hdots, \beta_k, \beta_{k+1}$ be the central bag decomposition of $G$ and let $w_i$ be the weight function on $\beta_i$ for $1 \leq i \leq k+1$. Let $\{v_1, v_2, v_3\}$, $G'$, $\beta_{k+1}''$, $Y$, and $G''$ be as in Theorem~\ref{thm:esd-summary}. By Theorem~\ref{thm:esd-summary}, we can find in polynomial time a semi-tame extended strip decomposition $(H_{k+1}, \eta_{k+1})$ of $(\beta_{k+1}'', \{v_1, v_2, v_3\})$ where every atom has weight at most $1-c$ (under $w_{k+1}$). 

Let $\beta_i' = (\beta_i \cup \{v_1, v_2, v_3\})$ for $1 \leq i \leq k$, let $\beta_{k+1}' = \beta_{k+1}''$, and let $Y_i = Y \cap \beta_i$. Let $C_i$ be the component of $\beta_i' \setminus Y_i$ that contains $\beta_{k+1}''$. Let $Q_{k+1} = Q_{k+1}' = \emptyset$. 
Let $Q_{i} = Y_{i} \cup Q_{i+1}'$ and let $Q_{i}' = N^{2(t+1)\Delta}[Q_{i}] \setminus C_{i}$. 
We now prove the following three statements inductively:
\begin{enumerate}[(i)] 

\item  $w_i(C) < 1-c$ for every component $C \neq C_i$ of $\beta_i' \setminus Q_i'$; 
\item $|Q_i'| \leq 3(1 + \Delta + \hdots + \Delta^t) \sum_{j = 1}^{k-i} (1 + \Delta + \hdots + \Delta^{2(t+1)\Delta})^j$ for all $k \leq i \leq 0$; and
\item we can obtain in polynomial time a semi-tame extended strip decomposition $(H_i, \eta_i)$ of $(C_i, \{v_1, v_2, v_3\})$ where every atom has weight at most $1-c$ (under $w_i$).
\end{enumerate} 
Since $Q_{k+1}' = \emptyset$, it follows that $C_{k+1} = \beta_{k+1}'$, so the base case is true. 

Assume the statement holds for $i+1$. Note that $Q_{i+1} \subseteq \beta_{i+1}'$. By Lemma~\ref{lemma:helper-combinatorial} (plugging in $Q = (Y_i \cup Q_{i+1}')\cap \beta_{i+1}'$), it follows that for every component $C \neq C_i$ of $\beta_i' \setminus Q_i'$, either $w_i(C) < 1-c$ or $w_i(C) \leq w_{i+1}(C')$ for some component $C'$ of $\beta_{i+1}' \setminus Q_{i+1}'$. By the inductive hypothesis, $w_{i+1}(C') < 1-c$ for every component $C'$ of $\beta_{i+1}' \setminus Q_{i+1}'$. It follows that $w_i(C) < 1-c$ for every component $C \neq C_i$ of $\beta_i' \setminus Q_i'$. This proves (i). 
Next, note that $|Q_i'| \leq |N^{2(t+1)\Delta}[Q_{i-1}']| + |N^{2(t+1)\Delta}[Y]|$. Since $|Y| \leq 3(1 + \Delta + \hdots + \Delta^t)$, this proves (ii). 
Finally, since $C_{i+1}$ is a component of $\beta_{i+1}' \setminus Q_{i+1}'$, $Q_{i+1}' \subseteq Q_i'$, and $C_{i+1} \cap Q_i' = \emptyset$, it follows that $C_{i+1} \subseteq C_i$. By the inductive hypothesis, we have a semi-tame extended strip decomposition $(H_{i+1}, \eta_{i+1})$ of $(C_{i+1}, \{v_1, v_2, v_3\})$. Therefore, by Lemma~\ref{lemma:esd-induction}, we can compute in polynomial time a semi-tame extended strip decomposition $(H_i, \eta_i)$ of $(C_i, \{v_1, v_2, v_3\})$ such that every atom has weight at most $1-c$ (under $w_i$). This proves (iii) and completes the induction.

Now, we have a set $Q_0'$ of size at most $z:=3(1 + \Delta + \hdots + \Delta^t) \sum_{j = 1}^{k} (1 + \Delta + \hdots + \Delta^{2(t+1)\Delta})^j$ with $Y \subseteq Q_0'$ and components $C_1, \hdots, C_m$ of $G' \setminus Q_0'$, such that $\beta_{k+1}'' \subseteq C_1$ and $w(C_i) < 1-c$ for all $2 \leq i \leq m$. We also obtain a semi-tame extended strip decomposition $(H_0, \eta_0)$ of $(C_1, \{v_1, v_2, v_3\})$ such that every atom has weight at most $1-c$ (under $w$). Let $(H', \eta')$ be the extended strip decomposition of $C_1 \setminus \{v_1, v_2, v_3\}$ given by deleting the vertices $v_1, v_2, v_3$ where they appear in $(H_0, \eta_0)$. Let $X = Q_0'$ and let $(H, \eta)$ be the extended strip decomposition of $G - X$ given by adding $m-1$ isolated vertices $c_2, \hdots, c_m$ to $H'$ and extending $\eta'$ to $\eta$ by setting $\eta(x) = \eta'(x)$ for all $x$ in the domain of $\eta'$, and $\eta(c_i) = V(C_i)$ for $2 \leq i \leq m$. Since every atom of $(H_0, \eta_0)$ has weight at most $1-c$ (under $w$) and $w(C_i) < 1-c$ for $2 \leq i \leq m$, it follows that every atom of $(H, \eta)$ has weight at most $1-c$ (under $w$). Since $c = 1 - \frac{1}{10\Delta}$ and $w(v) = \frac{1}{n}$ for all $v \in V(G)$, we conclude that every atom of $(H, \eta)$ has at most $\frac{1}{10\Delta} \cdot n$ vertices. Finally, since $(H_0, \eta_0)$ is semi-tame, it holds that every $\eta'(e)$ is nonempty, and so $|V(H')| \leq |C_1 \setminus \{v_1, v_2, v_3\}|$. Therefore, $|V(H)| \leq n$. This completes the proof. 
\end{proof}

\subsection{Solving Maximum Weight Independent Set} \label{sec:algo}
In this section we use Theorem~\ref{thm:combinatorial} to prove Theorem~\ref{thm:algorithm}.

\thmalgo*

Our proof follows the idea of Chudnovsky et al.~\cite{DBLP:conf/soda/ChudnovskyPPT20,DBLP:journals/corr/abs-1907-04585}, but our starting point is Theorem~\ref{thm:combinatorial}, instead of a similar statement obtained by Chudnovsky et al.
We aim to reduce the problem of finding a maximum-weight independent set in a graph given with an extended strip decomposition to the problem of finding a maximum-weight matching in an auxiliary graph.
The algorithm itself is exactly the same as the one of Chudnovsky et al., but we present it for the sake of completeness.

Let us start with introducing some more terminology. We point out that some notions are defined in a slightly different way than in~\cite{DBLP:conf/soda/ChudnovskyPPT20,DBLP:journals/corr/abs-1907-04585}. This is to make them consistent with the rest of our paper.

Let $G$ be a graph and let $(H,\eta)$ be an extended strip decomposition of $G$.
Recall that we have distinguished some sets called atoms (edge atoms, vertex atoms, and triangle atoms).

Let us now define some larger sets, called \emph{particles}.
First, for each $v \in V(H)$, each edge $uv \in E(H)$, and each triangle $uvw \in T(H)$ we have particles that are equal to the corresponding atoms
\begin{align*}
A_{v} = &\eta(v), \\
A_{uv}^{\perp} = & \eta(uv) \setminus \left(\eta(uv,u) \cup \eta(uv,v) \right),\\
A_{uvw} = &\eta(uvw).
\end{align*}
Additionally, for each $uv \in E(H)$ we distinguish three more particles related to $uv$.
\begin{align*}
A_{uv}^{u} = & \eta(u) \cup \eta(uv) \setminus \eta(uv,v), \\
A_{uv}^{v} = & \eta(v) \cup \eta(uv) \setminus \eta(uv,u), \\
A_{uv}^{uv} = & \eta(u) \cup \eta(v) \cup \eta(uv) \cup \bigcup_{w ~:~ uvw \in T(H)} \eta(uvw).
\end{align*}

The following lemma encapsulates the reduction by Chudnovsky et al.~\cite{DBLP:conf/soda/ChudnovskyPPT20,DBLP:journals/corr/abs-1907-04585}.
\begin{lemma}[\cite{DBLP:conf/soda/ChudnovskyPPT20,DBLP:journals/corr/abs-1907-04585}]\label{lem:matching}
Let $\varsigma \in [0,1]$ be a a real.
Let $G$ be an $n$-vertex graph equipped with a weight function $\wei: V(G) \to \N$.
Suppose that $G$ is given along with an extended strip decomposition $(H,\eta)$, where $H$ has $N$ vertices.

\noindent Let $I_0 \subseteq V(G)$ be a fixed independent set in $G$.
Furthermore, assume that for each particle $A$ of $(H,\eta)$ we are given an independent set $I(A)$ in $G[A]$, such that $\wei(I(A)) \geq \varsigma \cdot \wei(I_0 \cap A)$.
Then in time polynomial in $n+N$ we can compute an independent set $I$ in $G$, such that $\wei(I) \geq \varsigma \cdot \wei(I_0)$.
\end{lemma}
\begin{proof}
Let $H'$ be the graph obtained from $H$ as follows:
For each edge $uv$ of $H$ we add an additional vertex $x_{uv}$, adjacent to $u$ and $v$.
We also define a weight function $\wei'$ on \emph{edges} of $H'$. For each $uv \in E(H)$, we define:
\begin{align*}
\wei'(x_{uv}u) 	= & \wei(I(A_{uv}^u)) - \wei(I(A_u)) - \wei(I(A_{uv}^{\perp})) \\
\wei'(x_{uv}v) 	= & \wei(I(A_{uv}^v)) - \wei(I(A_v)) - \wei(I(A_{uv}^{\perp})) \\
\wei'(uv) 		= & \wei(I(A_{uv}^{uv})) - \wei(I(A_u)) - \wei(I(A_v)) - \wei(I(A_{uv}^{\perp})) - \sum_{w ~:~uvw \in T(H)} \wei(I(A_{uvw})).
\end{align*}
Let $M$ be a matching in $H'$.
Let $\cA(M)$ be the family of particles of $(H,\eta)$ constructed as follows:
\begin{itemize}
\itemsep0em
\item for every $uv \in E(H) \cap M$, insert $A_{uv}^{uv}$ into $\cA(M)$,
\item for every $x_{uv}u \in M \setminus E(H)$, insert $A_{uv}^{u}$ into $\cA(M)$,
\item for every $uv \in E(H)$ such that $uv, x_{uv}u, x_{uv}v \notin M$, insert $A_{uv}^{\perp}$ into $\cA(M)$,
\item for every $v \in V(H)$ such that no edge containing $v$ is in $M$, insert $A_{v}$ into $\cA(M)$,
\item for every $uvw \in T(H)$ such that $uv,uw,vw \notin M$, insert $A_{uvw}$ into $\cA(M)$.
\end{itemize}

The following claim binds independent sets in $G$ with matchings in $H'$.
For its proof we refer the reader to~\cite[Sections 3.3 and 3.4]{DBLP:journals/corr/abs-1907-04585}.

\begin{claim}
If $M$ is a maximum-weight matching in $H'$ (with respect to $\wei'$),
then $\bigcup_{A \in \cA(M)} I(A)$ is an independent set in $G$,
such that $\wei(I) \geq \varsigma \cdot \wei(I_0 )$.
\end{claim}

Now let us argue about the running time. 
Constructing $H'$ and computing $\wei'$ can clearly be performed in time polynomial in $n+N$.
The number of vertices of $H'$ is polynomial in $N$, so finding a maximum weight matching $M$ in $H'$ takes time polynomial in $N$. Finally, the solution $I$ is also computed in polynomial time. This completes the proof.
\end{proof}

Now we are ready to prove Theorem~\ref{thm:algorithm}.
\begin{proof}[Proof of Theorem~\ref{thm:algorithm}]
First, we can assume that $G$ is connected, as otherwise we can solve the problem component by component.
Next, note that we can safely assume that $\Delta \geq 3$, as otherwise $G$ is a path or a cycle and thus the problem is polynomial-time solvable.

Let $c$, $d$, $z$, and $p$ be given by Theorem~\ref{thm:combinatorial} for $t$ and $\Delta$.
Note that we can safely assume that $z \geq \log_2 \Delta$ and recall that we treat $c,d,z,p$ as constants.
Let $\delta = \zeta \cdot \max \{p, z, \frac{d}{\log_2(c+1)} \}$, where $\zeta$ is an absolute constant (independent on $t$ and $\Delta$), whose value follows from the reasoning below.

We will show by induction on $n$ that a maximum-weight independent set in $n$-vertex $S_{t,t,t}$-free graph with maximum degree at most $\Delta$ can be found in time $\Oh(n^\delta)$.
If $n$ is bounded by a constant, then the problem can be solved by brute-force and thus the claim holds.
So from now on assume that $n$ is sufficiently large and the claim holds for all graphs with fewer than $n$ vertices.
Let $F(n')$ be the running time of our algorithm on instances of size $n'$.

First, we check if $G$ has a $c$-balanced separator of size at most $d$. Clearly this can be done in time $\Oh(n^{d+2})$ by exhaustive enumeration of all vertex subsets $S$ of size at most $d$ and, for each of them, computing the connected components of $G-S$.

\medskip
\noindent\textbf{Case 1. $G$ has a $c$-balanced separator $S$ of size at most $d$.} The proof in this case is a standard divide and conquer procedure, we describe it for the sake of completeness.

\begin{claim}
The vertices of $G-S$ can be partitioned into sets $L$, $R$, each of size at most $\frac{c+1}{2}\; n$.
\end{claim}
\begin{claimproof}
Let $D_1,D_2,\ldots,D_r$ be the connected components of $G-S$, ordered in way that $|D_1| \geq |D_2| \geq \ldots \geq |D_r|$.
Recall that for each $D_i$ we have $|D_i|\leq cn$.

First, consider the case that $|D_1| \geq \frac{1-c}{2}\; n$. Note that then $|\bigcup_{i=2}^r D_i| \leq n-|D_1| \leq \frac{c+1}{2}\; n$.
Thus we can set $L=D_1$ and $R = \bigcup_{i=2}^r D_i$.

So let us assume that $|D_1| <\frac{1-c}{2}\; n$. Let $q$ be the maximum integer such that $|\bigcup_{i=1}^q D_i| < \frac{1-c}{2}\; n$.
Note that $q < r$ (here we use the fact that $n$ is large compared to $d$).
Furthermore, $|D_{q+1}| \leq |D_q| | < |\bigcup_{i=1}^q D_i| < \frac{1-c}{2}\; n$ and, by the maximality of $q$, we have $|\bigcup_{i=1}^{q+1} D_i| \geq \frac{1-c}{2}\; n$.
Setting $L = \bigcup_{i=1}^{q+1} D_i$ and $R = V(G)-S-L$,
we obtain that $\frac{1-c}{2}\; n < |L| \leq (1-c)n$ and thus $|R| \leq \frac{c+1}{2}\; n$.
\end{claimproof}
\medskip

Let $L$ and $R$ be as in the claim above.
We exhaustively guess the intersection of a fixed optimum solution with $S$. To this end, we enumerate all independent sets in $G[S]$, and for each such set $I_S$ we look for a maximum-weight independent set in $G$ whose intersection with $S$ is precisely $I_S$. This results in at most $2^{|S|} \leq 2^d$ branches.
In each branch we call the algorithm recursively for the graphs $G[L-N(I_S)]$ and $G[R-N(I_S)]$.
The algorithm returns the set with the largest weight among the solutions obtained in all branches.

The running time is bounded by the following recursive inequality:
\[
F(n) \leq \Oh(n^{d+2}) + 2^{d} \cdot 2 \cdot F \left( \frac{c+1}{2}\; n \right).
\]
A standard calculation shows that this can be upper-bounded by $\Oh(n^\delta)$, using the fact that $\delta \geq \zeta \; \frac{d}{\log_2(c+1)}$ for a large constant $\zeta$.

\medskip

\noindent\textbf{Case 2. $G$ has no $c$-balanced separator of size at most $d$.}
We apply Theorem~\ref{thm:combinatorial} to $G$ to obtain in time $\Oh(n^p)$ a set $X \subseteq V(G)$ of size at most $z$ and an extended strip decomposition $(H,\eta)$ of $G-X$, where the size of each atom is at most $\frac{1}{10\Delta}n$.  Additionally, the maximum degree of $H$ is at most $\Delta$ and the number of vertices of $H$ is at most $n$.

We exhaustively guess the intersection $I_X$ of a fixed optimum solution with the set $X$.
This results in at most $2^{z}$ branches. 

Suppose we consider one such branch for an independent set $I_X \subseteq X$.
Define $G' = G - X - N(I_X)$. Observe that now we need to find a maximum-weight independent set in $G'$.
We modify $(H,\eta)$ by removing from the sets $\eta$ all vertices from $N(I_X)$.
Note that this way we obtain an extended strip decomposition of $G'$; for simplicity we will keep calling it $(H, \eta)$.

Observe that each particle of $(H,\eta)$ consists of at most $2 + \Delta$ atoms (note that each edge of $H$ is in at most $\Delta-1$ triangles) and some subset of $\pot(v)$ for some $v \in V(H)$. Recall that $|\pot(v)| \leq 3\Delta/2$. Thus,  using the fact that $n$ is sufficiently large, we conclude that the size of each particle is at most
\[
\frac{2 + \Delta}{10\Delta} n + 3\Delta \leq 0.5n.
\]
Now note that each vertex is in at most $3\Delta$ particles, so the total size of all particles is at most $3\Delta \cdot n$.
Finally, for each particle $A$, the graph $G'[A]$ is $S_{t,t,t}$-free and has maximum degree at most $\Delta$.
Thus, by the inductive assumption, for each particle $A$ we can compute a maximum-weight independent set $I(A)$ in $G'[A]$ in time $F(|A|)=\Oh(|A|^\delta)$. This takes total time
\[
\Oh\left (\sum_{\text{particle } A} F(|A|) \right) = \Oh\left( \max_{\substack{x_1,x_2,\ldots,x_p \\ x_1+\ldots +x_p \leq 3\Delta \cdot n \\ x_i \leq 0.5n}  } \; \sum_{i=1}^p F(x_i) \right).
\]
As $F(x_i) = \Oh(x_i^\delta)$ for $\delta > 1$, a standard calculation shows that the above sum is maximized if the values of $x_i$ are as large as possible.
Thus the complexity of this step can be upper bounded by $6\Delta \cdot F(0.5 n)$. Now we can apply Lemma~\ref{lem:matching} for $\varsigma=1$ and $I_0$ being a maximum-weight independent set in $G$, in order to find a maximum-weight independent set in $G'$ in time $\Oh(n^{\gamma})$ for a constant $\gamma$ (for example we can take $\gamma = 4$ by using the classic algorithm of Edmonds~\cite{edmonds_1965}).

The total computation time is given by the recursive inequality:
\[
F(n) \leq \Oh(n^{p}) + 2^z \cdot 6\Delta \cdot F(0.5n) + \Oh(n^{\gamma}).
\]
Again, using standard calculation this can be upper-bounded by $\Oh(n^\delta)$, since $z \geq \log_2 \Delta$ and $\delta \geq \zeta \; \max\{ p, z\}$ for a large constant $\zeta$. This completes the proof.
\end{proof}

\section{Conclusion and open problems} 
We have shown that the MWIS problem is polynomial-time solvable in $S_{t,t,t}$-free graphs with bounded maximum degree.
An obvious question is whether one can remove the assumption that the maximum degree is bounded and prove that MWIS is polynomial-time solvable in $S_{t,t,t}$-free graphs. A natural first step would be to try to find a quasipolynomial-time algorithm, as it was the case for $P_t$-free graphs~\cite{gartlandpK}. However, these problems seem really challenging.

Our result suggests some more modest (and probably easier) research directions.
First, instead of removing the assumption that the degree is bounded, we could relax it and assume that our graph has bounded degeneracy.
Let us point out that a somewhat related result was obtained for $C_{> t}$-free graphs, i.e., graphs that do not contain any induced cycle with more than $t$ vertices. First, Chudnovsky et al.~\cite{DBLP:conf/soda/ChudnovskyPPT20,DBLP:journals/corr/abs-1907-04585} proved that such graphs have treewidth $\Oh(t \cdot \Delta)$, which, in particular, implies that if $\Delta$ is bounded by a constant, then MWIS is polynomial-time solvable. This was later improved by Gartland et al.~\cite{gartlandCk}, who proved that $C_{> t}$-free graphs of degeneracy at most $d$ have treewidth $(dt)^{\Oh(t)}$.

Next, recall that the running time of our algorithm on $n$-vertex $S_{t,t,t}$-free graphs with maximum degree $\Delta$ is $n^{f(t,\Delta)}$ for some function $f$. Thus we have shown that MWIS is in \textsf{XP}, when parameterized by the maximum degree $\Delta$ (we refer the reader to the texbook by Cygan et al.~\cite{Cyganetal} for introduction to parameterized complexity classes).
We think it is interesting to study whether the problem is \emph{fixed-parameter tractable} (\textsf{FPT}) with respect to $\Delta$, i.e., if there exists an algorithm with running time $f(\Delta,t) \cdot n^{g(t)}$, where $g$ does not depend on $\Delta$.

Finally, from the parameterized complexity point of view, we can ask about the complexity of the MIS problem under the natural parameterization by the solution size $k$.
Can we decide in time $f(k,t) \cdot n^{g(t)}$ if a given $n$-vertex $S_{t,t,t}$-free graph has an independent set of at least $k$ vertices?
This problem is also interesting for $P_t$-free graphs.

\paragraph*{Acknowledgment.}
We are grateful to Vadim Lozin for pointing out Corollary~\ref{cor:disconnected}.

\bibliographystyle{abbrv}
\bibliography{main}

\end{document}